\documentclass[sigconf,natbib=true,screen,anonymous=false]{acmart}
\usepackage{xspace,balance,tabularx,multirow}
\usepackage{flushend}
\usepackage{tikz}
\usepackage{pgfplots}
\usepgfplotslibrary{groupplots}
\pgfplotsset{compat=1.16}
\usetikzlibrary{patterns}
\usepackage{subfig}
\usepackage[ruled, vlined, linesnumbered]{algorithm2e}
\usepackage{xcolor}
\usepackage{colortbl}
\usepackage{bbold}
\SetKwComment{Comment}{$\triangleright$\ }{}
\usepackage{enumitem}
\usepackage{tablefootnote}
\usepackage{upgreek,textgreek}
\usepackage{pifont}%
\usepackage[noabbrev]{cleveref}
\usepackage{titlecaps}
\usepackage{lipsum}
\usepackage{makecell}

\pgfplotsset{every tick label/.append style={font=\tiny}}

\newlength{\starsize}
\newlength{\starspread}
\tikzset{starsize/.code={\setlength{\starsize}{#1}},
         starspread/.code={\setlength{\starspread}{#1}}}
\tikzset{starsize=1mm,
         starspread=3mm}
\pgfdeclarepatternformonly[\starspread,\starsize]%
  {my fivepointed stars}%
  {\pgfpointorigin}%
  {\pgfqpoint{\starspread}{\starspread}}%
  {\pgfqpoint{\starspread}{\starspread}}%
  {%
   \pgftransformshift{\pgfqpoint{\starsize}{\starsize}}
   \pgfpathmoveto{\pgfqpointpolar{18}{\starsize}}
   \pgfpathlineto{\pgfqpointpolar{162}{\starsize}}
   \pgfpathlineto{\pgfqpointpolar{306}{\starsize}}
   \pgfpathlineto{\pgfqpointpolar{90}{\starsize}}
   \pgfpathlineto{\pgfqpointpolar{234}{\starsize}}
   \pgfpathclose%
   \pgfusepath{fill}
  }

\newcommand{\argmax}[1]{\underset{#1}{\operatorname{arg}\,\operatorname{max}}\;}

\makeatletter
\newcommand*\bigcdot{\mathpalette\bigcdot@{.5}}
\newcommand*\bigcdot@[2]{\mathbin{\vcenter{\hbox{\scalebox{#2}{$\m@th#1\bullet$}}}}}
\makeatother

\newcommand{\stitle}[1]{\vspace*{0.5em}\noindent{\bf #1.\/}}

\newcommand{\U}{\mathcal{U}\xspace}
\newcommand{\V}{\mathcal{V}\xspace}

\newcommand{\G}{\mathcal{G}\xspace}
\newcommand{\N}{\mathcal{N}\xspace}
\newcommand{\EDG}{\mathcal{E}\xspace}
\newcommand{\C}{\mathcal{C}\xspace}
\newcommand{\LN}{\mathcal{L}\xspace}

\newcommand{\AM}{\mathbf{A}\xspace}

\newcommand{\DM}{\mathbf{D}\xspace}

\newcommand{\YM}{\mathbf{Y}\xspace}

\newcommand{\UM}{\mathbf{U}\xspace}
\newcommand{\VM}{\mathbf{V}\xspace}

\newcommand{\ZM}{\mathbf{Z}\xspace}

\newcommand{\uvec}{\vec{\mathbf{u}}\xspace}
\newcommand{\vvec}{\vec{\mathbf{v}}\xspace}

\newcommand{\fvec}{\vec{\mathbf{w}}\xspace}

\newcommand{\BM}{\mathbf{B}\xspace}

\newcommand{\algo}{\texttt{BACO}\xspace}

\newcommand{\bigzero}{\mbox{\normalfont\large\bfseries 0}}
\newcommand{\rvline}{\hspace*{-\arraycolsep}\vline\hspace*{-\arraycolsep}}

\newcommand{\eat}[1]{}

\newenvironment{customlegend}[1][]{%
    \begingroup
    \csname pgfplots@init@cleared@structures\endcsname
    \pgfplotsset{#1}%
}{%
    \csname pgfplots@createlegend\endcsname
    \endgroup
}%

\def\addlegendimage{\csname pgfplots@addlegendimage\endcsname}

\makeatletter
\newcommand\footnoteref[1]{\protected@xdef\@thefnmark{\ref{#1}}\@footnotemark}
\makeatother

\let\oldnl\nl%
\newcommand{\nonl}{\renewcommand{\nl}{\let\nl\oldnl}}%

\SetKwComment{Comment}{/* }{ */}

\makeatletter 
\g@addto@macro{\@algocf@init}{\SetKwInOut{Parameter}{Parameters}} 
\makeatother

\definecolor{myred}{HTML}{fd7f6f}
\definecolor{myred_new}{HTML}{D8D8D8}
\definecolor{myred_new2}{HTML}{D7191C}
\definecolor{myblue}{HTML}{7eb0d5}
\definecolor{mygreen}{HTML}{b2e061}
\definecolor{mypurple}{HTML}{bd7ebe}
\definecolor{myorange}{HTML}{ffb55a}
\definecolor{myyellow}{HTML}{ffee65}
\definecolor{mypurple2}{HTML}{beb9db}
\definecolor{mypink}{HTML}{fdcce5}
\definecolor{mycyan}{HTML}{8bd3c7}

\definecolor{myblue2}{HTML}{115f9a}
\definecolor{myred2}{HTML}{c23728}

\definecolor{NSCcol1}{HTML}{1f77b4}
\definecolor{NSCcol2}{HTML}{aec7e8}
\definecolor{NSCcol3}{HTML}{ff7f0e}
\definecolor{NSCcol4}{HTML}{ffbb78}
\definecolor{NSCcol5}{HTML}{98df8a}

\AtBeginDocument{%
  \providecommand\BibTeX{{%
    \normalfont B\kern-0.5em{\scshape i\kern-0.25em b}\kern-0.8em\TeX}}}

\setcopyright{acmcopyright}
\copyrightyear{2018}
\acmYear{2018}
\acmDOI{XXXXXXX.XXXXXXX}
\acmISBN{978-1-4503-XXXX-X/18/06}

\begin{document}

\title{Balanced Co-Clustering of Users and Items for Embedding Table Compression in Recommender Systems}
\subtitle{Technical Report}
\author{Runhao Jiang}
\affiliation{%
  \institution{Hong Kong Baptist University}
  \country{Hong Kong SAR, China}
}
\email{csrhjiang@comp.hkbu.edu.hk}
\orcid{0009-0000-4841-9175}

\author{Renchi Yang}
\authornote{Corresponding Author}
\affiliation{%
  \institution{Hong Kong Baptist University}
  \country{Hong Kong SAR, China}
}
\email{renchi@hkbu.edu.hk}
\orcid{0000-0002-7284-3096}

\author{Donghao Wu}
\authornote{Work done while an intern at HKBU}
\affiliation{%
  \institution{The Chinese University of Hong Kong}
  \country{Shenzhen, China}
}
\email{donghaowu@link.cuhk.edu.cn}
\orcid{0009-0002-4389-5329}

\renewcommand{\shortauthors}{Runhao Jiang, Renchi Yang, \& Donghao Wu}

\begin{abstract}

Recommender systems have advanced markedly over the past decade by transforming each user/item into a dense embedding vector with deep learning models.
At industrial scale, embedding tables constituted by such vectors of all users/items demand a vast amount of parameters and impose heavy compute and memory overhead during training and inference,
hindering model deployment under resource constraints.
Existing solutions towards embedding compression either suffer from severely compromised recommendation accuracy or incur considerable computational costs.

To mitigate these issues, this paper presents \algo{}, a fast and effective framework for compressing embedding tables. 
Unlike traditional ID hashing, \algo{} is built on the idea of exploiting collaborative signals in user-item interactions for user and item groupings, such that similar users/items share the same embeddings in the codebook.
Specifically, we formulate a balanced co-clustering objective that maximizes intra-cluster connectivity while enforcing cluster-volume balance, 
and unify canonical graph clustering techniques into the framework through rigorous theoretical analyses.
To produce effective groupings while averting codebook collapse, \algo{} instantiates this framework with a principled weighting scheme for users and items, an efficient label propagation solver, as well as secondary user clusters.
Our extensive experiments comparing \algo{} against full models and 18 baselines over benchmark datasets demonstrate that \algo{} cuts embedding parameters by over 75\% with a drop of at most 1.85\% in recall, while surpassing the strongest baselines by being up to $346\times$ faster.

\end{abstract}

\ccsdesc[300]{Information systems~Clustering}
\ccsdesc[500]{Information systems~Recommender systems}
\ccsdesc[300]{Mathematics of computing~Graph algorithms}
\keywords{Recommender Systems, Embedding Table Compression, Co-clustering}

\maketitle

\section{Introduction}

In modern recommender systems, deep learning models~\cite{zhang2019deep} have become the go-to approach for recommendations, typically working with embedding tables that map each category feature value, e.g., the ID of a user or item, to a unique vector representation.
However, in industrial-scale scenarios, these embedding tables consume a vast amount of parameters, often reaching hundreds of GB or even TB, due to the need to represent billions of users/items~\cite{guo2021scalefreectr}.
For instance, embedding tables deployed in Baidu's advertising systems may occupy up to 10 TB of storage~\cite{xu2021agile}, while at Meta’s scale, i.e., 3 billion monthly active users worldwide~\cite{gupta2020architectural}, the user embedding table in a recommender model can easily reach 715 GB of space for 64-dimensional vectors in 32-bit floating point.
As a consequence, the expansion of embedding tables intensifies storage and hardware demands, thus impeding model scalability and deployment in real production environments~\cite{zhang2020model,coleman2023unified}.

In recent years, {\em embedding table compression} (ETC)~\cite{li2024embedding}, which seeks to reduce the sizes of embedding tables with minimal degradation in precision, has emerged as a popular choice in industry~\cite{wu2020saec,wu2025graphhash,tsang2023clustering,chen_clustered_2023} and is gaining increasing traction in recommender systems~\cite{zhang2020model,coleman2023unified,chen2018learning}. Most existing works towards ETC generally follow three compression paradigms: {\em pre-training}, {\em in-training}, and {\em post-training} compressions, where the former generates the mappings for users/items in embedding tables before training recommendation models, while the latter two strategies compress the full embeddings when they are partially or fully learned.
Despite being effective, the in-training/post-training scheme additionally introduces considerable training overhead, and still suffers from substantial memory footprint.

A common treatment for pre-training ETC is to hash the user/item IDs down to a set of buckets with a manageable size through hashing functions~\cite{coleman2023unified,ghaemmaghami2022learning,weinberger2009feature,zhang2020model}.
The users/items mapped to the same bucket will share the same embedding vectors in the reduced embedding tables (a.k.a. {\em codebooks}).
Although this methodology is computationally efficient, it relies on hashing that is essentially random, which will represent unrelated users/items by the same embeddings, and hence, engendering severe {\em embedding collisions} and performance degradation~\cite{ghaemmaghami2022learning,zhang2020model}.
Subsequent studies alleviate embedding collision issues through the (i) employment of multiple hash functions~\cite{shi2020compositional,zhang2020model}, (ii) incorporation of feature and frequency statistics~\cite{coleman2023unified,ghaemmaghami2022learning,Zhang2018LSH}, and (iii) {\em learned hashing} techniques~\cite{tsang2023clustering,liang2024lightweight}.
These methods either still incur subpar recommendation performance, or require additional user/item information or training costs.
Most importantly, the collaborative signals underlying the user-item interactions are largely overlooked for ETC. Very recently, \citet{wu2025graphhash} made an attempt to exploit the collaborative information of users/items by reframing the ETC task as a graph clustering problem. They simply apply the highly-efficient \texttt{Louvain} algorithm~\cite{blondel2008fast} over the user-item interaction graph to derive clusters of users/items as their buckets, which surprisingly yield remarkable improvements in recommendation quality. Unfortunately, \texttt{Louvain} suffers from an inherent drawback of {\em resolution limit}~\cite{kim2022abc}, leading to suboptimal results.

Inspired by the efficacy of leveraging collaborative signals, we make further inroads by presenting \algo{}, a fast and effective solution for ETC via \underline{BA}lanced \underline{CO}-clustering.
Specifically, we first reveal two factors, i.e., {\em intra-cluster connectivity} and {\em cluster-size balance}, that are crucial to clustering for ETC, through an empirical study, which in essence correspond to the embedding collision and {\em codebook collapse}~\cite{tsang2023clustering}, respectively.
Based thereon, we develop a theoretically-grounded balanced co-clustering framework aiming at optimizing these two factors, and establish its theoretical connections to classic graph clustering algorithms.
Building on the framework, we propose a well-thought-out {\em hybrid weighting scheme} (HWS) for users and items, and develop a fast optimization solver through local greedy {\em label propagation}~\cite{raghavan2007near} to circumvent the resolution limit.
To account for the multiple and evolving interests of users, \algo{} further constructs {\em secondary clusters for users} (SCU) rapidly based upon the primary clusters, thereby overcoming the representation limitation of each user without inducing additional space overhead.
Our comprehensive empirical evaluations and ablation studies manifest that (i) our \algo{} can consistently achieve markedly superior recommendation quality over 18 ETC baselines, while often being highly efficient, and (ii) our proposed HWS and SCU strategies can conspicuously enhance recommendation accuracies even when working with other clustering baselines.

\section{Related Work}

We review existing studies on ETC in the sequel and defer reviews on co-clustering and bipartite graph clustering to Appendix~\ref{sec:add-rw}.

\stitle{Hashing-based Methods}
Hashing methods offer a straightforward and efficient means of compressing embedding tables by mapping IDs into a smaller index space using hash functions. For instance, the basic hashing method~\cite{weinberger2009feature} simply employs random function to achieve compression. Although efficient, this approach may introduce collisions. 
~\cite{Zhang2018LSH} increases the probability that colliding objects are similar through \texttt{LSH}, while ~\cite{shi2020compositional,zhang2020model} employ double hashing combined with other techniques to mitigate collisions.
In contrast, \texttt{ROBE}~\cite{desai2022ROBE} employs a more flexible indexing mechanism within compositional embeddings.
Although these hash methods are generally efficient, their weak association with the data makes it difficult to maintain accuracy.

\stitle{Vector Quantization}
Vector quantization (VQ) maps original embeddings to their most similar meta-embeddings, thereby bridging hash representations and the original data~\cite{zhangsurveyitem26}. \texttt{Saec}~\cite{wu2020saec} and \texttt{MGQE}~\cite{kang2020learning} utilize feature frequency to guide the quantization of embeddings. Furthermore, \texttt{xLightFM}~\cite{jiang2021xlightfm} proposes allocating various numbers of embeddings to different feature categories. 
Within multi-codebook frameworks, \texttt{LightRec}~\cite{lian2020lightrec} and \texttt{LISA}~\cite{wu2021LISA} aim to differentiate codebooks using distinct mechanisms. In contrast to earlier methods that focus on low-dimensional dense embeddings, recent work \texttt{CompresSAE}~\cite{kasalicky2025CompresSAE} introduces a radically different strategy by converting original embeddings into high-dimensional but sparse representations, combining expressive capacity with compression.
Vector quantization introduces significant computational and memory overhead during training, primarily due to the nearest neighbor search and the requirement to retain original embeddings.

\stitle{Decomposition-based Methods}
Decomposition methods perform soft selection, where each feature aggregates all codebook embeddings using a real-valued index vector.
\texttt{DHE}~\cite{kang2021DHE} decomposes the embedding table via hash functions and neural networks, whereas \texttt{ANT}~\cite{liang2020ANT} uses anchor vectors and sparse transformation. Recently, tensor train decomposition (TTD) has been employed to enhance compression by expressing multidimensional data as a product of smaller tensors~\cite{WangQinyong2020LLRec,yin2021TTRec,XiaXin2022STTD}. However, achieving such high compression ratios entails additional computational costs during decomposition and lookup, potentially slowing down inference.

\section{Preliminaries}

\subsection{Symbols and Terminology}\label{sec:symbol}
\stitle{Symbols} 
We model the interactions between a user set $\U=\{u_1,u_2,\ldots,u_{|\U|}\}$ and an item set $\V=\{v_1,v_2,\ldots,v_{|\V|}\}$ as a {\em bipartite network} $\G=(\U\cup \V,\EDG)$, where the edge set $\EDG\subseteq \U\times \V$ consists of all the user-item interactions.
The neighbors of user $u_i$ (resp. item $v_j$) are represented by $\N(u_i)$ (resp. $\N(v_j)$). 
We denote by $\BM\in \{0,1\}^{|\U|\times |\V|}$ the bi-adjacency matrix of $\G$, where $\BM_{i,j}=1$ if user $u_i$ has interacted with item $v_j$, i.e., $(u_i,v_j)\in \EDG$, and 0 otherwise. We use $\footnotesize\textstyle {\AM} = \begin{pmatrix}
\bigzero
  & \rvline & \begin{matrix} \BM \end{matrix} \\
\hline
   \hspace{0.7em}\BM^{\top} & \rvline &
  \begin{matrix}
\bigzero
  \end{matrix}
\end{pmatrix}\in \{0,1\}^{(|\U|+|\V|)\times (|\U|+|\V|)}$ to represent the complete adjacency matrix of $\G$.
The degree of user $u_i$ (resp. item $v_j$) is symbolized by $d(u_i)$ (resp. $d(v_j)$) and the diagonal degree matrix is denoted by $\DM \in \mathbb{R}^{(|\U|+|\V|)\times (|\U|+|\V|)}$.

Throughout this paper, we use $\uvec_i$ (resp. $\vvec_j$) to denote the embedding vector of user $u_i$ (item $v_j$).
For notation convenience, $\UM \in\mathbb{R}^{|\U|\times d}$ and $\VM \in\mathbb{R}^{|\V|\times d}$ represent the embedding vectors of users in $\U$ and items in $\V$, respectively, wherein the $i$-th row of $\UM$ (resp. $\VM$) is $\UM_i=\uvec_i$ (resp. $\VM_i=\vvec_i$). In recommender systems, $\UM$ and $\VM$ are often referred to as embedding tables of users and items, respectively.
Given a cluster $\C_k$ containing users and items, we define $\U_k=\C_k\cap\U$ and $\V_k=\C_k\cap\V$.
We use $[K]$ to denote the set of integers $\{1,2,\ldots,K\}$. Table~\ref{tbl:symbol} lists frequently used notions in our paper.

\stitle{Modularity and CPM}
The {\em modularity}~\cite{newman2004finding} quantifies the goodness of a particular division of a network. Formally, given a unipartite network $\G$ and clusters $\{\C_1,\C_2,\ldots,\C_K\}$, the modularity is defined by $\footnotesize\textstyle \frac{1}{|\EDG|}\sum_{k=1}^K{\left( s_k - \gamma\cdot\frac{\left(\sum_{v_i\in \C_k}{d(v_i)}\right)^2}{|\EDG|}\right)}$,
where 
$\gamma>0$ stands for a resolution parameter~\cite{reichardt2006statistical}, and $s_k=|\{(v_i,v_j)\in \EDG|v_i,v_j\in \C_k\}|$ is the actual number of edges within cluster $\C_k$. 
Intuitively speaking, ${\left(\sum_{v_i\in \C_k}{d(v_i)}\right)^2}/{|\EDG|}$ can be interpreted as the expected number of edges in cluster $\C_k$.
In particular, higher (resp. lower) resolutions lead to more (resp. fewer) communities, and a larger modularity value indicates a better partitioning of the network $\G$.
~\citet{barber2007modularity} extended modularity to bipartite networks, formulating {\em bipartite modularity} as follows: 
\begin{small}
\begin{equation}\label{eq:bmod}
\frac{1}{|\EDG|}\sum_{k=1}^K{\left( s_k - \gamma\cdot\frac{\sigma^{(u)}_k\cdot \sigma^{(v)}_k}{|\EDG|}\right)},
\end{equation}
\end{small}
where $\sigma^{(u)}_k=\sum_{u_i\in \U_k}{d(u_i)}$ (resp. $\sigma^{(v)}_k=\sum_{v_i\in \V_k}{d(v_i)}$) denotes the sums of degrees of nodes in $\C_k$ and set $\U$ (resp. $\V$).

\begin{table}[!t]
\centering
\renewcommand{\arraystretch}{0.9}
\begin{footnotesize}
\caption{Frequently used symbols.}\vspace{-3mm} \label{tbl:symbol}
\resizebox{\columnwidth}{!}{%
\begin{tabular}{|p{0.51in}|p{2.5in}|}
\hline
{\bf Symbol}  &  {\bf Description}\\
\hline
$\U,\V, \EDG$   & The user set, item set and edge set.\\ \hline
$\scriptstyle |\U|, |\V|, |\EDG|$ & The number of users, items, and all the edges.\\ \hline
$u_i, v_i$ & A user in $\U$, and a item in $\V$.\\ \hline
$\G, \BM$ & The bipartite network and its bi-adjacency matrix.\\ \hline
$\AM$ & The complete adjacency matrix of $\G$. \\ \hline
$\DM$ & The diagonal degree matrix of $\AM$. \\ \hline
$\UM,\VM$ & The embedding tables of users and items. \\ \hline
$\uvec_i, \vvec_j$ & The embedding vector of user $u_i$ and item $v_j$.\\ \hline
$\C_k$ & A cluster containing users and items. \\ \hline
$\U_k,\V_k$ & The user set and the item set of cluster $\C_k$. \\ \hline
$[K]$ & The set of integers $\{1,2,\ldots,K\}$.\\ \hline
$\ZM^{(u)},\ZM^{(v)}$ & The compressed embedding table for users and items. \\ \hline
$K^{(u)},K^{(v)}$ & The numbers of clusters for users and items.\\ \hline
$\YM^{(u)},\YM^{(v)}$ & The sketching matrices of users and items. \\ \hline
$w_i^{(u)},w_j^{(v)}$ & The weights for user $u_i$ and item $v_j$.\\ \hline
$\fvec$ & The vector containing the weights of users and items.\\ \hline
$W^{(u)},W^{(v)}$ & The total weight of the users and items in $\G$.\\ \hline
$\gamma$ & The coefficients for objective term in Eq.~\eqref{eq:overall-obj} \\ \hline
\end{tabular}%
}
\end{footnotesize}
\vspace{-2ex}
\end{table}

An alternative quality function like the modularity is called the {\em Constant Potts Model} (CPM)~\cite{traag2011narrow}, whose mathematical formulation is defined by $\sum_{k=1}^K{\left( s_k - \gamma \cdot \tbinom{|\C_k|}{2}\right)}$,
which can also be extended to bipartite networks by substituting $|\U_k|\cdot|\V_k|$ for $\tbinom{|\C_k|}{2}$, i.e.,
$\sum_{k=1}^K{\left( s_k - \gamma \cdot |\U_k|\cdot|\V_k|\right)}$.
In the literature~\cite{traag2019louvain}, both quality functions can be efficiently and effectively optimized by the prominent \texttt{Louvain} algorithm invented by~\citet{blondel2008fast}.

\subsection{Problem Statement}\label{sec:problem}
Let $\ZM^{(u)}\in \mathbb{R}^{K^{(u)}\times d}$ (resp. $\ZM^{(v)}\in \mathbb{R}^{K^{(v)}\times d}$) be the {\em codebook} (a.k.a. compressed embedding table) for users (resp. items), where $K^{(u)}\ll |\U|$, $K^{(v)}\ll |\V|$, and each row vector stands for an embedding vector.
The {\em embedding table compression}\footnote{We follow the ``pre-training'' compression setting for parameter efficiency~\cite{tsang2023clustering}.} (ETC) in recommender systems aims to find a mapping $f$: $\U\rightarrow [K^{(u)}]$ and a mapping $h$: $\V\rightarrow [K^{(v)}]$ such that each user $u_i$ (resp. item $v_j$) can be represented by an embedding vector $\ZM^{(u)}_{f(i)}$ (resp. $\ZM^{(v)}_{h(j)}$).

The foregoing mappings $f$ and $h$ can be equivalently represented by {\em sketching matrices} $\YM^{(u)} \in \{0,1\}^{|\U|\times K^{(u)}}$ and $\YM^{(v)} \in \{0,1\}^{|\V|\times K^{(v)}}$, respectively. Each row $\YM^{(u)}_i$ (resp. $\YM^{(v)}_j$) therein is a one-hot vector, usually called ``sketch of user $i$ (resp. item $v_j$)''~\cite{tsang2023clustering}. Accordingly, $\UM = \YM^{(u)}\ZM^{(u)}$ and $\VM = \YM^{(v)}\ZM^{(v)}$. 

Notice that given sketching matrices $\YM^{(u)}$ and $\YM^{(v)}$, recommendation models only need to learn codebook $\ZM^{(u)}$ and $\ZM^{(v)}$ in the course of training.
As such, the space overhead originally incurred for storing the embedding tables $\UM$ and $\VM$ in the recommender systems is reduced from $O((|\U|+|\V|)\cdot d)$ to $O(|\U|+|\V| + (K^{(u)}+K^{(v)})\cdot d)$ for mappings and codebooks. Given a space budget $B$ for the codebooks, i.e., the total number of embedding vectors in codebooks for users and items, the ETC task ensures $K^{(u)}+K^{(v)}\le B$. 

In model training, given a user $u_i$ (resp. item $v_j$), the compressed embedding $\ZM^{(u)}_{f(i)}$ (resp. $\ZM^{(v)}_{h(j)}$) is first retrieved before entering into training. For example, the predicted rating for user $u_i$ and item $v_j$ is obtained as $\hat{y}_{i,j}=\langle \ZM^{(u)}_{f(i)},\ZM^{(v)}_{h(j)} \rangle$, and the loss function like $\mathcal{L}_{\textnormal{BPR}} = -\sum_{u_i\in \U}{\sum_{v_k\in \N(u)}{\sum_{v_j\notin \N(u)}}{\ln{\sigma(\hat{y}_{i,k}-\hat{y}_{i,j})}}} + \lambda \cdot \left(\|\UM\|^2 + \|\VM\|^2\right)$, is computed with $\UM = \YM^{(u)}\ZM^{(u)}$ and $\VM = \YM^{(v)}\ZM^{(v)}$ accordingly.

\subsection{Construction of Sketching Matrices}\label{sec:SM}
There are generally three categories of approaches for constructing the sketching matrices $\YM^{(u)}$ and $\YM^{(v)}$: {\em random sketching}, {\em learned sketching}, and {\em clustering-based sketching}.

Random sketching simply employs hashing functions as mappings $f$: $\U\rightarrow [K^{(u)}]$ and $h$: $\V\rightarrow [K^{(v)}]$. Each user or item is assigned to a random sketch (a random one-hot vector)~\cite{weinberger2009feature}. 
Subsequent works~\cite{shi2020compositional,zhang2020model,desai2022ROBE,tito2017hash} resort to compositional embeddings, which basically apply two or more hashing functions to the IDs of users and items, such that each user or item is associated with two or more embeddings in the codebook, which can be combined via summation, element-wise multiplication, or concatenation, etc. 
Note that this approach yields denser sketching matrices as each row $\YM^{(u)}_i$ and $\YM^{(v)}_j$ comprise multiple non-zero entries.

Instead of generating random sketches, the learned sketching~\cite{tsang2023clustering,liang2024lightweight} methodology learn the sketch of each user or item by training models to reconstruct sketching matrices.
After every few epochs, sketching matrices $\YM^{(u)}$ and $\YM^{(v)}$ are regenerated upon current embeddings $\YM^{(u)}\ZM^{(u)}$ and $\YM^{(v)}\ZM^{(v)}$ to achieve a better fit.

By interpreting sketching matrices $\YM^{(u)}$ and $\YM^{(v)}$ as node-cluster indicator matrices, e.g., $\YM^{(u)}_{i,k}=1$ if $u_i$ belongs to the $k$-th user cluster and $0$ otherwise, the rationale of clustering-based sketching is to group ``similar'' users in $\U$ and items in $\V$ into $K^{(u)}$ and $K^{(v)}$ clusters, respectively.
\texttt{GraphHash}~\cite{wu2025graphhash} implements this idea by clustering users and items over the bipartite network structure $\G$, attaining markedly superior effectiveness over random sketching and significantly higher efficiency than learned sketching.

\begin{figure}[!t]
\centering
\begin{small}
\begin{tikzpicture}
   \hspace{3mm}\begin{customlegend}[
        legend entries={\textsf{Random},\textsf{Frequency},\textsf{GraphHash}, \textsf{LP},\textsf{EBMD},\textsf{SCC},\textsf{SBC},\algo{}}, 
        legend columns=4,
        legend style={at={(0.45,1.25)},anchor=north,draw=none,font=\footnotesize,column sep=0.25cm}]
        \addlegendimage{only marks, mark=*, mark size=2pt, color=mypurple2}   
        \addlegendimage{only marks, mark=*, mark size=2pt, color=NSCcol3}  
        \addlegendimage{only marks, mark=*, mark size=2pt, color=NSCcol1} 
         \addlegendimage{only marks, mark=*, mark size=2pt, color=cyan}   
        \addlegendimage{only marks, mark=*, mark size=2pt, color=mypurple}     
        \addlegendimage{only marks, mark=*, mark size=2pt, color=myblue}     
        \addlegendimage{only marks, mark=*, mark size=2pt, color=myred}    
        \addlegendimage{only marks, mark=*, mark size=2pt, color=mycyan}
    \end{customlegend}
\end{tikzpicture}
\\[-8pt]
\vspace{-1mm}

\hspace{-2mm}
\subfloat[{\em \#cross-cluster links}]{
\begin{tikzpicture}[scale=1]
\begin{axis}[
    height=3.2cm,
    width=3.7cm,
    enlarge x limits=0.15,
    ylabel={\em Recall@20},
    xtick={0.2,0.3,0.4,0.5},
    xticklabels={$0.02$,$0.03$,$0.04$,$0.05$},
    xmin=0.18,
    xmax=0.49,
    ymin=8,
    ymax=17,
    ytick={8,10,12,14,16},
    xticklabel style = {font=\scriptsize},
    yticklabel style = {font=\scriptsize},
    every axis y label/.style={at={(current axis.north west)},right=5mm,above=0mm},
    legend style={draw=none, at={(1.02,1.02)},anchor=north west,cells={anchor=west},font=\scriptsize},
    legend image code/.code={ \draw [#1] (0cm,-0.1cm) rectangle (0.3cm,0.15cm);},
    ]

\addplot [only marks, mark=*, mark size=2pt, color=mypurple2] coordinates {(0.48764,9.214)}; 
\addplot [only marks, mark=*, mark size=2pt, color=NSCcol3] coordinates {(0.48761,8.574)}; 
\addplot [only marks, mark=*, mark size=2pt, color=NSCcol1] coordinates {(0.41615,15.325)}; 
\addplot [only marks, mark=*, mark size=2pt, color=mypurple] coordinates {(0.47449,10.313)};
\addplot [only marks, mark=*, mark size=2pt, color=cyan] coordinates {(0.17117,11.749)};
\addplot [only marks, mark=*, mark size=2pt, color=myblue] coordinates {(0.44415,13.792)}; 
\addplot [only marks, mark=*, mark size=2pt, color=myred] coordinates {(0.48881,10.458)}; 
\addplot [only marks, mark=*, mark size=2pt, color=mycyan] coordinates {(0.331479,16.393)}; 
\addplot [ domain=0.16:0.50, samples=100, color=NSCcol1, very thick, dotted ] { -845.28*x^3 + 600.96*x^2 - 107.15*x +16.717 };
\draw [red, thick] (axis cs:0.482,9.65) ellipse [x radius=0.04, y radius=1.55];
\end{axis}
\end{tikzpicture}\hspace{0mm}\label{fig:CorssEdge}%
}%
\subfloat[{\em Gini coefficient (users)}]{
\begin{tikzpicture}[scale=1]
\begin{axis}[
    height=3.2cm,
    width=3.8cm,
    enlarge x limits=0.15,
    ylabel={\em Recall@20},
    xmin=0,
    xmax=1.0,
    xtick={0,0.2,0.4,0.6,0.8,1},
    ymin=8,
    ymax=17,
    ytick={8,10,12,14,16},
    xticklabel style = {font=\scriptsize},
    yticklabel style = {font=\scriptsize},
    every axis y label/.style={at={(current axis.north west)},right=5mm,above=0mm},
    legend style={draw=none, at={(1.02,1.02)},anchor=north west,cells={anchor=west},font=\scriptsize},
    legend image code/.code={ \draw [#1] (0cm,-0.1cm) rectangle (0.3cm,0.15cm);},
    ]

\addplot [only marks, mark=*, mark size=2pt, color=mypurple2] coordinates {(0.0254,9.214)};    
\addplot [only marks, mark=*, mark size=2pt, color=NSCcol3] coordinates {(0.430,8.574)}; 
\addplot [only marks, mark=*, mark size=2pt, color=NSCcol1] coordinates {(0.3824,15.325)}; 
\addplot [only marks, mark=*, mark size=2pt, color=mypurple] coordinates {(0.3831,10.313)};
\addplot [only marks, mark=*, mark size=2pt, color=cyan] coordinates {(0.8852,11.749)} ; 
\addplot [only marks, mark=*, mark size=2pt, color=myblue] coordinates {(0.6473,13.792)}; 
\addplot [only marks, mark=*, mark size=2pt, color=myred] coordinates {(0.6305,10.458)};  
\addplot [only marks, mark=*, mark size=2pt, color=mycyan] coordinates {(0.70574,16.393)}; 
\draw [red, thick] (axis cs:0.35, 9.68) ellipse [x radius=0.4, y radius=1.7];
\end{axis}
\end{tikzpicture}\hspace{4mm}\label{fig:Gini-user}%
}%
\subfloat[{\em Gini coefficient (items)}]{
\begin{tikzpicture}[scale=1]
\begin{axis}[
    height=3.2cm,
    width=3.8cm,
    enlarge x limits=0.15,
    ylabel={\em Recall@20},
    xmin=0,
    xmax=1.0,
    xtick={0,0.2,0.4,0.6,0.8,1},
    ymin=8,
    ymax=17,
    ytick={8,10,12,14,16},
    xticklabel style = {font=\scriptsize},
    yticklabel style = {font=\scriptsize},
    every axis y label/.style={at={(current axis.north west)},right=5mm,above=0mm},
    legend style={draw=none, at={(1.02,1.02)},anchor=north west,cells={anchor=west},font=\scriptsize},
    legend image code/.code={ \draw [#1] (0cm,-0.1cm) rectangle (0.3cm,0.15cm);},
    ]

\addplot [only marks, mark=*, mark size=2pt, color=mypurple2] coordinates {(0.0078,9.214)};     
\addplot [only marks, mark=*, mark size=2pt, color=NSCcol3] coordinates {(0.4478,8.574)}; 
\addplot [only marks, mark=*, mark size=2pt, color=NSCcol1] coordinates {(0.3753,15.325)}; 
\addplot [only marks, mark=*, mark size=2pt, color=mypurple] coordinates {(0.3961,10.313)};
\addplot [only marks, mark=*, mark size=2pt, color=cyan] coordinates {(0.8870,11.749)} ; 
\addplot [only marks, mark=*, mark size=2pt, color=myblue] coordinates {(0.6438,13.792)} ; 
\addplot [only marks, mark=*, mark size=2pt, color=myred] coordinates {(0.6536,10.458)}; 
\addplot [only marks, mark=*, mark size=2pt, color=mycyan] coordinates {(0.6954,16.393)}; 
\draw [red, thick] (axis cs:0.35, 9.68) ellipse [x radius=0.41, y radius=1.7];
\end{axis}
\end{tikzpicture}\hspace{1mm}\label{fig:Gini-item}%
}%

\vspace{-3mm}
\end{small}
\caption{Recommendation performance (Recall@20) v.s. averaged \#cross-cluster links and Gini coefficients on {\em Gowalla}} \label{fig:preliminary} 
\vspace{-3ex}
\end{figure}
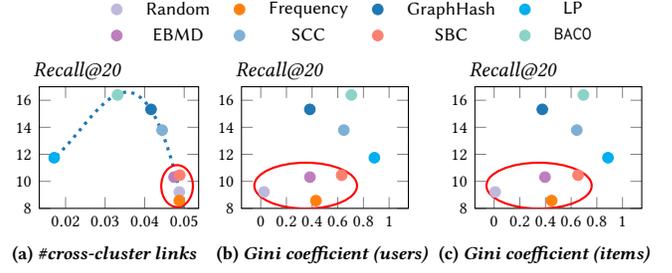

\stitle{Empirical Study} To gain deeper insights, we empirically study the correlation between final recommendation performance and the characteristics of sketching matrices produced by various methods from the perspective of clustering. Specifically, we generate sketching matrices using hashing-based methods (\texttt{Random} and \texttt{Frequency}) and clustering-based methods (\texttt{GraphHash}, \texttt{LP}, \texttt{EMBD}, \texttt{SCC}, \texttt{SBC}, and our proposed \texttt{BACO}) on the {\em Gowalla} dataset. Figure~\ref{fig:preliminary} plots the recall@20 scores ($y$-axis) achieved by the above methods, and their respective average number of cross-cluster links and the Gini coefficients of their user and item cluster sizes ($x$-axis). Particularly, fewer cross-cluster links connections suggest that related users/items are well grouped, implying rare embedding collisions. 
A low Gini coefficient signals that clusters are nearly equal in size, indicating codebook embeddings are used evenly across users or items (i.e., mild codebook collapse). 

From Figure~\ref{fig:CorssEdge}, when inter-cluster connectivity is excessively high (see methods circled in red), recommendation performance suffers regardless of whether codebook usage (measured by the Gini coefficient) is uniform or uneven (Figures~\ref{fig:Gini-user} and ~\ref{fig:Gini-item}).
In contrast, \texttt{LP} generates clusters of users and items with a minimal number of cross-cluster links, but causes high Gini coefficients that indicate severe embedding collision and codebook collapse, and hence, result in poor recommendations.
\texttt{GraphHash} and \algo{} obtain much better recall scores by balancing these two factors.
From the foregoing observations, we can derive the following insight: {\em the underlying clusters of high-quality sketching matrices should minimize inter-cluster connectivity while avoiding imbalanced cluster sizes.}

\section{Methodology}
In this section, we first present our framework \algo{} for co-clustering users and items in \S~\ref{sec:framework} and unify classic graph clustering algorithms into this framework through rigorous theoretical analyses in \S~\ref{sec:theoretical}. Building on this framework, we propose the {\em hybrid weighting scheme} (HWS) specialized for users and items in the context of e-commerce in \S~\ref{sec:node-weight}, followed by constructing co-clusters via an efficient optimization solver in \S~\ref{sec:LP}.
Lastly, to alleviate the embedding collision and codebook collapse issues in co-clusters, \S~\ref{sec:post} introduces a simple but effective approach to create {\em secondary clusters for users} (SCU).

\subsection{The Balanced Co-Clustering Framework}\label{sec:framework}
\begin{table}[!t]
\centering
\caption{The unified balanced co-clustering framework.}\label{tbl:framework}
\vspace{-2ex}
\renewcommand{\arraystretch}{0.8}
\small
\begin{tabular}{lcccc}
\toprule
{\bf Method} & $\gamma$ & $w^{(u)}_i$ & $w^{(v)}_j$ & Opt. solver \\
\midrule
\texttt{Louvain} (Modularity)~\cite{blondel2008fast} & $>0$ & $\frac{d(u_i)}{\sqrt{|\EDG|}}$ & $\frac{d(v_j)}{\sqrt{|\EDG|}}$ & \texttt{Louvain} \\
\texttt{Louvain} (CPM) & $>0$ & $1$ & $1$ & \texttt{Louvain} \\
\texttt{Leiden} (Modularity)~\cite{traag2019louvain} & $>0$ & $\frac{d(u_i)}{\sqrt{|\EDG|}}$ & $\frac{d(v_j)}{\sqrt{|\EDG|}}$ & \texttt{Louvain} \\
\texttt{Leiden} (CPM)~\cite{traag2019louvain} & $>0$ & $1$ &  $1$ & \texttt{Louvain} \\
\texttt{LP}~\cite{raghavan2007near} & 0 & - & - & \texttt{LP} \\
\texttt{LPAb}~\cite{barber2009detecting} & $>0$ & $\frac{d(u_i)}{\sqrt{|\EDG|}}$ & $\frac{d(v_j)}{\sqrt{|\EDG|}}$  & \texttt{LP} \\ 
\texttt{SCC}~\cite{dhillon2001co} & 0 & - & - & Eigensolver \\ \hline
Our \algo{} & $>0$ & $\frac{d(u_i)}{\sqrt{|\EDG|}}$ & $\frac{1}{\sqrt{|\V|}}$ & \algo{} \\
\bottomrule
\end{tabular}
\end{table}

We frame the construction of sketching matrices $\YM^{(u)}$ and $\YM^{(v)}$ as co-clustering users and items in $\G$ into $K$ disjoint co-clusters $\{\C_1,\ldots,\C_K\}$. Particularly, we represent the node-cluster memberships by an indicator matrix $\YM\in \{0,1\}^{(|\U|+|\V|)\times K}$, where $\YM_{i,k}=1$ if user $u_i\in \C_k$, $\YM_{|\U|+j,k}=1$ if item $v_j\in \C_k$, and $0$ otherwise. Note that $\YM$ can be converted into $\YM^{(u)}$ and $\YM^{(v)}$
by mapping the clusters of users and items into consecutive column indices, respectively.

\stitle{Maximizing Intra-cluster Connectivity} The first objective in our \algo{} aims to identify clusters $\{\C_1,\ldots,\C_K\}$ such that the {\em intra-cluster connectivity} is maximized. 
More concretely, users and items that are densely connected via interactions in $\EDG$ should fall into the same clusters since it connotes a high correlation or similarity between users (resp. items). This leads to the maximization of the number of connections within the same clusters, i.e., intra-cluster connectivity, which can be formulated as follows:
\begin{equation}\label{eq:max-intra}
\max_{\C_1,\ldots,\C_K}\sum_{k=1}^K\sum_{u_i, v_j\in \C_k}{\BM_{i,j}}.
\end{equation}
The objective can be further transformed into a trace maximization problem with indicator matrix $\YM$ and adjacency matrix $\AM$ of $\G$:
\begin{align}
& \max_{\C_1,\ldots,\C_K} \sum_{k=1}^K\sum_{u_i, v_j\in \C_k}{\BM_{i,j}} \Leftrightarrow \max_{\C_1,\ldots,\C_K}\sum_{k=1}^K\sum_{u_i, v_j\in \C_k}{\AM_{i,|\U|+j}} \notag \\
& \Leftrightarrow \max_{\YM}\sum_{k=1}^K\sum_{u_i\in \U}\sum_{v_j\in \V}{\YM_{i,k}\cdot \AM_{i,|\U|+j}\cdot \YM_{|\U|+j,k}} \notag \\
& \Leftrightarrow \max_{\YM}\sum_{k=1}^K{(\YM^\top\AM\YM)_{k,k}} \Leftrightarrow \max_{\YM} \mathsf{Trace}(\YM^\top\AM\YM).\label{eq:trace-max}
\end{align}
Such an optimization task is equivalent to the mincut problem for graph partitioning in the literature \cite{stoer1997simple}, which is proved NP-hard. 

\stitle{Weighted Exclusive Lasso for Size Balance}
As remarked earlier in \S~\ref{sec:SM}, in overly imbalanced clusters, substantial users/items with low relevance are likely to be grouped together, which yields severe embedding collision and codebook collapse, and thus, degrades performance.
As a remedy, in \algo{}, we additionally include a weighted exclusive lasso to balance the sizes/volumes of the $K$ clusters.

Specifically, instead of treating all users and items equally, we assign a weight $w^{(u)}_i$ (resp. $w^{(v)}_j$) to each user $u_i$ (item $v_j$).
Let $W^{(u)}=\sum_{u_i\in \U}{{w^{(u)}_i}}$ and $W^{(v)}=\sum_{v_j\in \V}{{w^{(v)}_j}}$ be the total weight of the users and items in $\G$, respectively. Accordingly, we define the volume of a cluster $\C_k$ as 
\begin{equation}\label{eq:vol}
\textsf{vol}(C_k) = \sum_{u_i\in \U_k}{w^{(u)}_i} + \sum_{v_j\in \V_k}{{w^{(v)}_j}},
\end{equation}
which is a summation of the weights of the users and items therein.

Ideally, the volumes of all clusters should be comparable to each other, i.e., minimizing the weighted exclusive lasso:
\begin{gather}
\min_{C_1,C_2,\ldots,C_K}{\sum_{k=1}^K\left(\textsf{vol}(C_k)-\frac{W^{(u)}+W^{(v)}}{K}\right)^2}
\Leftrightarrow \min_{C_1,C_2,\ldots,C_K}{\|\fvec^\top\YM\|_2^2} \notag \\
\Leftrightarrow  \min_{C_1,C_2,\ldots,C_K}{\mathtt{Trace}(\YM^\top\fvec\fvec^\top\YM)}.\label{eq:lasso}
\end{gather}
where $\fvec \in\mathbb{R}^{|\U|+|\V|}$ is a vector containing the weights of users and items, e.g., $\fvec_i=w^{(u)}_i$ for user $u_i\in \U$ and $\fvec_{|\U|+j}=w^{(v)}_j$ for item $v_j\in \V$.

\stitle{Overall Objective} Combining the foregoing intra-cluster connectivity maximization (Eq.~\eqref{eq:trace-max}) and weighted exclusive lasso minimization (Eq.~\eqref{eq:lasso}) leads to the overall objective of \algo{}:
\begin{equation}\label{eq:overall-obj}
\max_{\YM}\mathtt{Trace}(\YM^\top\AM\YM) - \gamma\cdot \mathtt{Trace}(\YM^\top\fvec\fvec^\top\YM),
\end{equation}
where the coefficient $\gamma$ can be used to strengthen the weight of the cluster size balance term.

\subsection{Theoretical Insights}\label{sec:theoretical}
Next, we conduct a theoretical investigation of \algo{} framework, so as to (i) establish its connections to several classic clustering algorithms over bipartite graphs, and (ii) uncover its indirect impact on the downstream user and item embeddings $\UM$ and $\VM$.

\stitle{Connections to Other Clustering Methods}
We define the Kronecker $\delta$ function as follows:
\begin{equation}
\delta(u_i, v_j) = \begin{cases} 1 & \text{if } u_i,v_j \in \C_k, \text{ i.e., } \YM_i = \YM_{|\U|+j}, \\
0 & \text{otherwise}. \end{cases}
\end{equation}
The first intra-cluster connectivity term in Eq.~\eqref{eq:overall-obj} can be equivalently expressed by
\begin{align}
& \mathtt{Trace}(\YM^\top\AM\YM) = \sum_{k=1}^K{(\YM^\top\AM\YM)_{k,k}} 
 = \sum_{k=1}^K\sum_{u_i, v_j\in \C_k}{\BM_{i,j}} \notag\\
 & = \sum_{u_i\in \U}\sum_{v_j\in \V}{\BM_{i,j}\cdot \delta(u_i,v_j)}  = \sum_{u_i\in \U} \sum_{v_j\in \N(u_i)}{\BM_{i,j}\cdot \delta(u_i,v_j)}. \label{eq:trace-to-sum}
\end{align}
Analogously, the second term $-\gamma\cdot\mathtt{Trace}(\YM^\top\fvec\fvec^\top\YM)$ is equal to
$$-\gamma\cdot\sum_{u_i\in \U} \sum_{v_j\in \V}{\fvec_i\cdot\fvec_{|\U|+j}\cdot \delta(u_i,v_j)}.$$
Along this line, our overall objective in Eq.~\eqref{eq:overall-obj} can be rewritten as
\begin{align}
\max_{\YM}\sum_{u_i\in \U} \sum_{v_j\in \V}{\left(\BM_{i,j}-\gamma\cdot w^{(u)}_i\cdot w^{(v)}_j\right)\cdot \delta(u_i,v_j)}.\label{eq:overall-obj-new}
\end{align}

Together with the above formulation, we are enabled to unify existing clustering algorithms tailored to bipartite graphs including the modularity- and CPM-based \texttt{Louvain}~\cite{blondel2008fast} and \texttt{Leiden}~\cite{traag2019louvain}, Label Propagation (\texttt{LP})~\cite{raghavan2007near}, Spectral Co-clustering (\texttt{SCC})~\cite{dhillon2001co}, and \texttt{LPAb}~\cite{barber2009detecting} into our \algo{} framework, as summarized in Table~\ref{tbl:framework}, based on their respective optimization functions theoretically analyzed in the following lemmata.

\begin{lemma}\label{lem:louvain}
Given bipartite graph $\G$, the adoption of \texttt{Louvain}, \texttt{Leiden}, or \texttt{LPAb} over $\G$ can maximize the following form of the bipartite modularity: $\sum_{u_i\in \U} \sum_{v_j\in \V}{\left(\BM_{i,j}-\gamma\cdot \frac{d(u_i)\cdot d(v_j)}{|\EDG|}\right)\cdot \delta(u_i,v_j)}$, or the following form of the CPM: $\sum_{u_i\in \U} \sum_{v_j\in \V}{\left(\BM_{i,j}-\gamma\right)\cdot \delta(u_i,v_j)}$.
\end{lemma}

\begin{lemma}\label{lem:LP}
The optimization objective of both \texttt{LP} and \texttt{SCC} can be expressed as $\max_{\YM}\sum_{u_i\in \U}\sum_{v_j\in \V}{\BM_{i,j}\cdot \delta(u_i,v_j)}$.
\end{lemma}

It can be observed that the key differences of these algorithms lie in the choices of parameter $\gamma$, weights of users and items, as well as the solver used for optimization.

\stitle{Impact on User and Item Embeddings}
The overall objecive in Eq.~\eqref{eq:overall-obj-new} can be rewritten as
\begin{equation*}
\max_{\C_1,\ldots,\C_K}\sum_{k=1}^K\sum_{u_i\in \U_k}\sum_{v_j\in \V_k}{\BM_{i,j}-\gamma\cdot w^{(u)}_i\cdot w^{(v)}_j}.
\end{equation*}
From the perspective of each user $u_i$, if we are given fixed item clusters $\{\V_1,\ldots,\V_K\}$, $u_i$ is assigned to the user cluster $\U_k$ when its total {\em penalized} interaction weight with items in $\V_k$ is minimized, i.e., 
\begin{equation}\label{eq:kappa}
\kappa(u_i)=\argmax{1\le \ell \le K}{\sum_{v_j\in \V_\ell}{\BM_{i,j}-\gamma\cdot w^{(u)}_i\cdot w^{(v)}_j}}=k.
\end{equation}
Accordingly, users $u_i$ and $u_l$ are assigned to the same cluster $\U_k$ if and only if $\kappa(u_i)=\kappa(u_l)=k$.
Recall that in the resulting user and item embeddings $\UM$ and $\VM$, users or items from the same clusters will share the same embedding vector in the codebook $\ZM^{(u)}$ and $\ZM^{(v)}$ (see \S~\ref{sec:problem}). The above Eq.~\eqref{eq:kappa} implies that users with similar interaction patterns with $K$ sets of items are likely to have the same embedding vectors.
A similar conclusion can be made for item embeddings when the user clusters are fixed.

\begin{algorithm}[!t]
\caption{The Basic \algo{} Algorithm}\label{alg:LP}
\KwIn{Bipartite graph $\G$, space budget $B$, integer $T$, and coefficient $\gamma$}
\KwOut{Sketching matrices $\YM^{(u)}$ and $\YM^{(v)}$}
\tcc{Initializing labels of users and items}
$\ell(u_i)\gets i \ \forall{u_i\in \U}$ and $\ell(v_j)\gets |\U|+j \ \forall{v_j\in \V}$\;
$t\gets 0$\;
\tcc{Updating labels of users and items}
\While{$K^{(u)}+K^{(v)}>B$ and $t<T$}{
\For{$x_i\in \G$}{
$\LN_i \gets \{\ell(y_j)|y_j\in \N(x_i)\cup x_i\}$\;
\For{$k \in \LN_i$}{
\If{$x_i\in \U$}{
Compute $p(k)$ according to Eq.~\eqref{eq:p_l_u}\;
}\Else{
Compute $p(k)$ according to Eq.~\eqref{eq:p_l_v}\;
}
}
$\ell(x_i) \gets \argmax{k \in \LN_i}{p(k)}$\;
}
$t\gets t+1$\;
}
\tcc{Re-labeling users and items}
$\ell^{(u)}: \{\ell(u_i)|u_i\in \U\}\rightarrow [K^{(u)}]$\;
$f(u_i)\gets \ell^{(u)}(\ell(u_i))\ \forall{u_i\in \U}$\;
$\ell^{(v)}: \{\ell(v_j)|v_j\in \V\}\rightarrow [K^{(v)}]$\;
$h(v_j)\gets \ell^{(v)}(\ell(v_j))\ \forall{v_j\in \V}$\;
\tcc{Constructing sketching matrices}
$\YM^{(u)}_{i,f(u_i)} \gets 1\ \forall{u_i\in \U}$ and $\YM^{(v)}_{j,h(v_j)} \gets 1\ \forall{v_j\in \V}$\;
\end{algorithm}

\subsection{Hybrid Weighting Scheme}\label{sec:node-weight}
As pinpointed in Table~\ref{tbl:framework}, previous clustering methods unified into the \algo{} framework all adopt the same weighting functions for users and items in $\G$, either a degree-related function $\frac{d(x)}{\sqrt{|\EDG|}}$ used in the bipartite modularity or a constant (e.g., $1$) used in the CPM.
However, this simple strategy neglects the intrinsic distinction of users and items in the context of recommendation systems, and hence, causes flawed cluster assignments.

More precisely, recall that in Eq.~\eqref{eq:kappa}, user $u_i$ is assigned to cluster $\C_k$ if $\sum_{v_j\in \C_k}{\BM_{i,j}-\gamma\cdot w^{(u)}_i\cdot w^{(v)}_j}$ is maximal. 
When the weight $w^{(v)}_j=\frac{d(v_j)}{\sqrt{|\EDG|}}$ is adopted for items, user $u_i$ will be prone to join the cluster of its interacted items with low degrees (i.e., unpopular items) due to the penalty term $-\gamma\cdot w^{(u)}_i\cdot \frac{d(v_j)}{\sqrt{|\EDG|}}$, which is counterintuitive. In \algo{}, we view all items interacted by user $u_i$ equally instead, leading to the following weight for item $v_j$  
\begin{equation}
w^{(v)}_j = \frac{1}{\sqrt{|\V|}}.
\end{equation}

On the other hand, item $v_j$ will be merged to cluster $\C_k$ with maximal $\sum_{u_i\in \C_k}{\BM_{i,j}-\gamma\cdot w^{(u)}_i\cdot w^{(v)}_j}$. In this case, item $v_j$ will prioritize the cluster of users with lower weights $w^{(u)}_i$.
Recall that in real-world scenarios, compared to those who interacted with plenty of items, the interaction with item $v_j$ from a user $u_i$ with fewer interactions conveys a stronger preference and interest. Intuitively, $v_j$ should be more likely to be grouped with $u_i$. Given this observation, we resort to the degree-related weight for users, as exemplified below:
\begin{equation}
w^{(u)}_i = \frac{d(u_i)}{\sqrt{\sum_{u_\ell \in \U}{d(u_\ell)}}} = \frac{d(u_i)}{\sqrt{|\EDG|}}.
\end{equation}

\subsection{The Optimization Algorithm}\label{sec:LP}

The pseudo-code of our optimization solver for \algo{} is illustrated in Algorithm~\ref{alg:LP}. Given the bipartite graph $\G=(\U,\V,\EDG)$, the space budget $B$ for codebooks, the maximum number $T$ of iterations, and a coefficient $\gamma$, \algo{} begins by assigning a {\em unique} cluster label $\ell(u_i)\gets i$ (resp. $\ell(v_j)\gets |\U|+j$) to each user $u_i$ (item $v_j$) (Line 1) and reset the iteration count (Line 2).  
Afterwards, Algorithm~\ref{alg:LP} starts an iterative procedure (Lines 3-12) to update the labels of users and items until the total number of labels for users and items $K^{(u)}+K^{(v)}$ reaches the space budget $B$, where $K^{(u)}$ and $K^{(v)}$ are defined as
\begin{equation*}
K^{(u)} \gets |\{\ell(u_i)|u_i\in \U\}|,\quad K^{(v)} \gets |\{\ell(v_j)|v_j\in \V\}|.
\end{equation*}

In each iteration, \algo{} picks each user or item from $\G$ and updates its label to the one with the highest likelihood among all the labels of its adjacent neighbors, akin to the {\em label propagation} scheme~\cite{raghavan2007near}. Specifically, for each node $x_i\in \G$, \algo{} first collects its candidate label set $\LN_i$ from the neighbors and $x_i$ at Line 5. Then, for every label $k\in \LN_i$, we calculate its likelihood by
\begin{equation}\label{eq:p_l_u}
p(k) \gets  \sum_{v_j\in \N(x_i)\cap \C_k}{\BM_{i,j}} - \gamma \sum_{v_j\in \V_k}{w^{(u)}_i\cdot w^{(v)}_j}
\end{equation}
if $x_i$ is a user in $\U$, and by
\begin{equation}\label{eq:p_l_v}
p(k) \gets \sum_{u_j\in \N(x_i)\cap \C_k}{\BM_{j,i}} - \gamma \sum_{u_j\in \U_k}{w^{(u)}_j\cdot w^{(v)}_{i}}
\end{equation}
if it is an item in $\V$ (Lines 6-10).
Next, at Line 11, the new label of $x_i$ is set to the one with the maximum likelihood. 
Note that the first term in the likelihood $p(k)$ measures the intra-cluster connectivity of $x_i$ to its adjacent items (resp. users) with label $k$, while the second term considers the penalty $\gamma\cdot w^{(u)}_j\cdot w^{(v)}_{i}$ for all items (resp. users) with label $k$.
In essence, this label updating rule is to maximize the objective Eq.~\eqref{eq:overall-obj-new} {\em locally} for each target node $x_i$ in a greedy fashion.

After that, Algorithm~\ref{alg:LP} maps the cluster labels of users and items to consecutive integers $[K^{(u)}]$ and $[K^{(v)}]$, respectively, followed by relabeling all the users and items accordingly (Lines 13-16). Based thereon, we can construct sketching matrices $\YM^{(u)}$ and $\YM^{(v)}$ at Line 17 with the new labels.

\stitle{Remark} Although \texttt{Louvain} can be applied to optimize our objective in Eq.~\eqref{eq:overall-obj-new} with minor changes, its optimization strategy tends to merge small clusters, leading to the {\em resolution limit} issue~\cite{kim2022abc} on bipartite networks. In other words, \texttt{Louvain} will greedily group nodes with low connectivity for higher overall modularity. The empirical cluster size distributions reported in Figure~\ref{fig:cluster-size-dist}
confirm this problem, and Table~\ref{tab:weighting} manifests its detrimental effects on recommendation performance.

\begin{algorithm}[!t]
\caption{The Complete \algo{} Algorithm}\label{alg:BACO}
{\nonl Input, output, and Lines 1-2 are the same as Algorithm~\ref{alg:LP}\\}
\setcounter{AlgoLine}{2}
\textbf{while} $K^{(u)}+K^{(v)}>B^\prime$ and $t<T$ \textbf{do}\\
{\nonl Lines 4-17 are the same as Algorithm~\ref{alg:LP}\;}
\setcounter{AlgoLine}{17}
Rerun Lines 4-10 in Algorithm~\ref{alg:LP} for each $x_i\in \U$\;
$\ell^{(u)}_{\text{scu}}: \{\ell(u_i)|u_i\in \U\}\rightarrow [K^{(u)}]$\;
$f_{\text{scu}}(u_i)\gets \ell^{(u)}_{\text{scu}}(\ell(u_i))\ \forall{u_i\in \U}$\;
$\YM^{(u)}_{i,f_{\text{scu}}(u_i)} \gets 1\ \forall{u_i\in \U}$\;
\end{algorithm}

\subsection{Generating Secondary Clusters for Users}\label{sec:post}
Although our basic \algo{} in Algorithm~\ref{alg:LP} can achieve an effective balance and optimization of two terms in Eq.~\eqref{eq:overall-obj}, and thus, produce good divisions of users and items, as empirically validated in \S~\ref{sec:exp-weighting} and \S~\ref{sec:exp-SCU}, it suffers from a fundamental deficiency of representing each user via exactly one embedding in codebooks, whilst e-commerce users have multiple, evolving interests and often share taste similarities with various user groups.

To remedy this issue while maintaining parameter efficiency, we propose to assign each user to two clusters, which demands a small space of $O(|\U|)$ for additional user sketches, leading to a new space budget for the codebooks: $B^\prime\gets\frac{B\cdot d - |\U|}{d}$.
Accordingly, the loop condition at Line 3 in Algorithm~\ref{alg:LP} will be changed to $K^{(u)}+K^{(v)}>B^\prime$. Next, the task is to construct the secondary clusters for users.
Specifically, as displayed in Algorithm~\ref{alg:BACO}, we employ a simple but effective SCU strategy, which only repeats Lines 4-10 in Algorithm~\ref{alg:LP} for each user in $\U$ for once to generate updated cluster labels as their secondary clusters (Line 18). 
\algo{} then maps users to $K^{(u)}$ consecutive integers based on the secondary cluster labels, i.e., $f_{\text{scu}}: \U\rightarrow [K^{(u)}]$ at Lines 19-20.
Eventually, for each user $u_i\in \U$, we set the entry $\YM^{(u)}_{i,f_{\text{scu}}(u_i)}$ that corresponds to its secondary cluster in its sketch to $1$ (Line 21).

\subsection{Complexity Analysis}
In Algorithm~\ref{alg:LP}, each iteration accesses the labels of neighbors in $O(|\EDG|)$ time at Lines 3-4. 
Since only one label is updated at each step (Line 11), at most two clusters are affected. 
By maintaining global sums of cluster weights and updating them whenever label changes, the computation of Eq.~\eqref{eq:p_l_u} or Eq.\eqref{eq:p_l_v} can be done in $O(1)$. 
To ensure termination over disconnected $\G$, we impose a maximum number $T$ (typically $8$) of iterations on the iterative procedure (Lines 4-11). In practice, the label propagation process (Lines 3-12) can usually converge and terminate rapidly~\cite{raghavan2007near}, i.e., the actual number of iterations is less than $T$.\footnote{As validated by our empirical results in Appendix~\ref{sec:param-set}.
}
Therefore, Algorithm~\ref{alg:LP} takes $O(T|\EDG|)$ time.  
Similarly, obtaining the secondary cluster labels requires only $O(|\EDG|)$ time. At Lines 13-17 and 19-21, clusters are mapped to sets of consecutive integers, which takes $O(|\U| + |\V|)$ time. Overall, the time complexity of the complete \algo{} is $O(|\U| + |\V| + T|\EDG|)$. 
The space complexity is bounded by $O(|\U| + |\V| + |\EDG|)$ due to the storage of $\G$ and clusters.

\section{Experiments}\label{sec:exp}
This section experimentally evaluates \algo{} against 18 ETC baselines in recommendations, and conducts related ablation studies and component analyses to answer the following research questions:
\begin{itemize}[leftmargin=*]
    \item \textbf{RQ1}: How does \algo{} perform in terms of accuracy compared to full models and ETC baselines in recommendation tasks? (\S~\ref{exp-performance})
    \item \textbf{RQ2}: How does \algo{} perform in terms of computation efficiency and parameter reduction compared to ETC baselines? (\S~\ref{sec:exp-efficiency})
    \item \textbf{RQ3}: How effectively do the HWS and SCU enhance the performance of \algo{}? (\S~\ref{sec:exp-weighting} and \S~\ref{sec:exp-SCU})
\end{itemize}
All experiments are conducted on a Linux machine equipped with an NVIDIA Ampere A100 GPU (80 GB), AMD EPYC 7513 CPUs (2.6 GHz), and 1 TB RAM. The codes and datasets are made publicly available at \url{https://github.com/HKBU-LAGAS/BACO}.

\subsection{Experimental Setup}

\begin{table}[!t]
\renewcommand{\arraystretch}{0.9}
\centering
\caption{Statistics of datasets used in experiments.}
\label{tab:datasets-retrieval}
\vspace{-2ex}
\begin{small}
\begin{tabular}{c|cccc}
\hline
{\bf Dataset} & {\bf \#Users} & {\bf \#Items} & {\bf \#Interactions} & {\bf Density} \\ \hline
Beauty & 22,363	& 12,101 & 198,502 & 0.073\% \\
Gowalla & 29,858 & 40,981 & 1,027,370 & 0.084\% \\
Yelp2018 & 31,668 & 38,048 & 1,561,406 & 0.130\% \\
AmazonBook & 52,643 & 91,599 & 2,984,108 & 0.062\% \\
\hline
\end{tabular}
\end{small}
\vspace{-1ex}
\end{table}

\stitle{Datasets}
Table~\ref{tab:datasets-retrieval} summarizes the four datasets used in the retrieval task, including {\em Beauty}~\cite{wang2022towards}, {\em Gowalla}~\cite{cho2011friendship}, {\em Yelp2018}~\cite{wu2025graphhash}, {\em AmazonBook}~\cite{he2016ups}, each representing a specific recommendation domain and all publicly available. For each dataset, we randomly divide the data into 80\% for training, 10\% for validation, and 10\% for testing.

\stitle{Baselines}
For a thorough evaluation, we include 18 ETC baselines in the experiments, broadly categorized into three groups:
\begin{itemize}[leftmargin=*]
\item Hashing methods:
\texttt{Random}, \texttt{Frequency Hashing}~\cite{ghaemmaghami2022learning,zhang2020model}, \texttt{Double Hashing}~\cite{zhang2020model}, \texttt{Hybrid Hashing}~\cite{zhang2020model}, \texttt{LSH}~\cite{datar2004locality},
\texttt{CCE}~\cite{tsang2023clustering}, \texttt{LEGCF}~\cite{liang2024lightweight}. 

\item Graph clustering methods:
(\texttt{Double})\texttt{GraphHash}~\cite{wu2025graphhash},  \texttt{Leiden}~\cite{traag2019louvain}, \texttt{LP}~\cite{raghavan2007near}, \texttt{EBMD}~\cite{kim2022abc}, \texttt{infomap}~\cite{rosvall2008maps}, \texttt{BiM\-LPA}~\cite{taguchi2020bimlpa}, \texttt{BRIM}~\cite{platig2016bipartite}.

\item Co-clustering methods:
\texttt{SCC}~\cite{dhillon2001co}, \texttt{SBC}~\cite{kluger2003spectral}, \texttt{ITCC}~\cite{dhillon2003information}.
\end{itemize}

\stitle{Evaluation Protocol}
We evaluate top-$K$ item recommendation performance using two widely adopted ranking metrics: Recall@$K$ and NDCG@$K$, with $K = 20$ by default. The metrics are averaged over all users in the test set. For fair comparison, all methods are implemented with an identical experimental protocol, employing the classical \texttt{LightGCN}~\cite{he2020lightgcn} model as the backbone and the BPR loss function. Further experimental details are provided in Appendix~\ref{sec:add-exp-details}. 

\begin{table*}[!t]
\centering
\caption{Recommendation Performance ($k=20$). Best results highlighted in {\color{blue!30}blue} and runner-up \underline{underlined}.}
\label{tab:results}
\vspace{-2ex}
\begin{small}
\renewcommand{\arraystretch}{0.9}
\addtolength{\tabcolsep}{-0.25em}
\resizebox{\textwidth}{!}{
\begin{tabular}{c|*{4}{ccc}}
\toprule
\multirow{2}{*}{Method} & 
\multicolumn{3}{c}{Beauty} & 
\multicolumn{3}{c}{Gowalla} & 
\multicolumn{3}{c}{Yelp2018} & 
\multicolumn{3}{c}{AmazonBook} \\
 \cmidrule(lr){2-4} \cmidrule(lr){5-7} \cmidrule(lr){8-10} \cmidrule(lr){11-13}
\multicolumn{1}{c|}{} & 
\multicolumn{1}{c}{\#Params\textdownarrow} & \multicolumn{1}{c}{R@20\textuparrow} & \multicolumn{1}{c}{N@20\textuparrow} & 
\multicolumn{1}{c}{\#Params\textdownarrow} & \multicolumn{1}{c}{R@20\textuparrow} & \multicolumn{1}{c}{N@20\textuparrow} & 
\multicolumn{1}{c}{\#Params\textdownarrow} & \multicolumn{1}{c}{R@20\textuparrow} & \multicolumn{1}{c}{N@20\textuparrow} & \multicolumn{1}{c}{\#Params\textdownarrow} & \multicolumn{1}{c}{R@20\textuparrow} & \multicolumn{1}{c}{N@20\textuparrow} \\
\midrule

 Full Model
& 2.206M & 10.931 & 5.769
& 4.534M & 18.204 & 11.600 
& 4.462M & 8.867 & 5.526
& 9.231M & 8.416 & 5.487 
\\ \cline{1-13}

 \texttt{Random}
& 0.529M & 4.696 & 2.347 
& 0.742M & 9.214 & 5.909
& 0.977M & 4.894 & 3.114
& 1.255M & 2.485 & 1.704 \\

 \texttt{Frequency Hashing}~\cite{ghaemmaghami2022learning} 
& 0.529M & 4.367 & 2.194
& 0.742M & 8.574 & 5.480 
& 0.977M & 4.931 & 3.142
& 1.255M & 2.184 & 1.497  \\

 \texttt{Double Hashing}~\cite{zhang2020model}
& 0.529M & 4.132 & 2.080 
& 0.742M & 9.215 & 5.994  
& 0.977M & 5.293 & 3.366
& 1.255M & 2.398 & 1.687 \\

 \texttt{Hybrid Hashing}~\cite{zhang2020model}
& 0.529M & 4.564 & 2.307 
& 0.742M & 10.182 & 6.581 
& 0.977M & 6.090 & 3.884
& 1.255M & 2.800 & 1.884 \\

 \texttt{LSH}~\cite{datar2004locality}
& 1.049M & 5.566 & 2.787 
& 1.049M & 9.768 & 6.253 
& 1.049M & 4.707 & 2.952
& 2.097M & 3.214 & 2.167 \\ %

\texttt{CCE}~\cite{tsang2023clustering}
& 0.529M & 7.400 & 3.699
& 0.742M & 10.455 & 6.851 
& 0.977M & 5.505 & 3.538
& 1.255M & 3.233 & 2.124\\

 \texttt{LEGCF}~\cite{liang2024lightweight}
& 0.529M & 7.908 & 3.738
& 0.742M & 8.463 & 5.125
& 0.977M & 4.349 & 2.685
& OOM & - & - \\ \cline{1-13}

 \texttt{GraphHash}~\cite{wu2025graphhash}
& 0.529M & 8.167 & 4.272 
& 0.742M & 15.325 & 9.658
& 0.977M & 6.244 & 3.882
& 1.255M & 7.261 & 5.033 \\

 \texttt{DoubleGraphHash}~\cite{wu2025graphhash}
& 0.529M & 6.266 & 3.198
& 0.742M & 12.377 & 7.905 
& 0.977M & 5.134 & 3.236
& 1.255M & 5.194 & 3.501\\   

 \texttt{LP}~\cite{raghavan2007near}
& 0.530M & 7.883 & 4.088
& 0.743M & 11.749 & 7.464
& 0.977M & 5.423 & 3.391
& 1.255M & 3.483 & 2.275  \\  

 \texttt{Leiden}~\cite{traag2019louvain}
& 0.529M & 7.943 & 4.141 
& 0.743M & \underline{15.406} & \underline{9.738}
& 0.979M & 6.230 & 3.852
& 1.256M & \underline{7.426} & \underline{5.115} \\

 \texttt{EBMD}~\cite{kim2022abc}
& 0.536M & 8.095 & 4.037 
& 0.746M & 10.313 & 6.564
& 0.996M & 4.947 & 3.088
& 1.261M & 3.332 & 2.199\\ %

 \texttt{infomap}~\cite{rosvall2008maps}
& 0.025M & 4.614 & 2.252 
& 0.016M & 5.893 & 3.590 
& 0.037M & 4.335 & 2.688 
& 0.018M & 1.166 & 0.736 \\

 \texttt{BiMLPA}~\cite{taguchi2020bimlpa} 
& 0.001M & 3.765 & 1.708
& 0.007M & 6.397 & 3.920 
& 0.008M & 4.462 & 2.767
& 0.001M & 3.765 & 1.708  \\

  \texttt{BRIM}~\cite{platig2016bipartite}
& 0.003M & 8.253 & 4.013 
& 0.002M & 8.074 & 5.150 
& 0.001M & 4.483 & 2.789 
& 0.001M & 2.302 & 1.440 \\ \cline{1-13}

 \texttt{SCC}~\cite{dhillon2001co} 
& 0.529M & \underline{8.801} & \underline{4.476} 
& 0.743M & 13.792 & 8.817
& 0.977M & \underline{6.633} & \underline{4.124}
& 1.255M & 5.558 & 3.615\\ 

 \texttt{SBC}~\cite{kluger2003spectral} 
& 0.529M & 6.816 & 3.345
& 0.742M & 10.458 & 6.715
& 0.977M & 4.831 & 3.023
& 1.255M & 3.556 & 2.372\\

 \texttt{ITCC}~\cite{dhillon2003information}
& 0.529M & 4.812 & 2.437 
& 0.742M & 7.685 & 5.156 
& 0.977M & 4.375 & 2.844 
& 1.255M & 2.610 & 1.790  \\ \cline{1-13}

\algo{}
& 0.529M & \cellcolor{blue!30}9.079 & \cellcolor{blue!30}4.574
& 0.742M & \cellcolor{blue!30}16.692 & \cellcolor{blue!30}10.674
& 0.976M & \cellcolor{blue!30}7.510 & \cellcolor{blue!30}4.763 
& 1.255M & \cellcolor{blue!30}7.892 & \cellcolor{blue!30}5.240\\

 v.s. Baselines
& - & +0.278 & +0.098
& - & +1.286 & +0.936
& - & +0.877 & +0.639
& - & +0.466 & +0.125  \\

 v.s. Full Model
& -76.0\%& -1.852 & -1.195  
& -83.6\% & -1.512 & -0.926
& -78.1\% & -1.357 & -0.763
& -86.4\% & -0.524 & -0.247\\

\bottomrule
\end{tabular}
}
\end{small}
\end{table*}

\subsection{Recommendation Performance Evaluation}\label{exp-performance}

Table~\ref{tab:results} reports the mean Recall@20 and NDCG@20 performance of \algo{} and the baselines across four datasets, averaged over 5 independent runs using the best hyperparameters. The results under other top-$K$ (10, 50) values are quantitatively similar, and thus, are deferred to appendix. We can make the following observations.

Firstly, our proposed algorithm consistently outperforms all baselines across all tested datasets, with marked improvements. For instance, our algorithm surpasses the best baselines by considerable margins of $0.28\%$, $1.29\%$, $0.88\%$, and $0.47\%$ in terms of Recall@20 on {\em Beauty}, {\em Gowalla}, {\em Yelp2018}, and {\em AmazonBook}, respectively. Particularly, \algo{} attains an average improvement of $0.96\%$ in recall and $0.56\%$ in NDCG over state-of-the-art community detection methods, and offers significantly faster runtime. Additionally, \algo{} consistently maintains outstanding performance, while \texttt{SCC} and \texttt{Leiden} tend to fail on certain datasets. These results demonstrate the effectiveness of our \algo{} in exploiting and integrating collaborative information for recommendation across diverse contexts, enabled by its hybrid weighting schemes and refined secondary cluster design. 

Moreover, compared with simple hash methods (e.g., \texttt{Random} and \texttt{Frequency}), which rely solely on statistical probabilities, graph-based methods (e.g., \texttt{GraphHash} and \texttt{LP}) consistently demonstrate superior performance across most datasets because of their ability to incorporate graph structural information.
Unlike conventional clustering algorithms (e.g., $k$-means), most community detection methods (such as \texttt{Infomap}, \texttt{BiMLPA}, and \texttt{BRIM}) do not provide control over the cluster number, which leads to fewer parameters but inferior performance in our experiments. In addition, \texttt{SCC}, as a classical co-clustering method, serves as the strongest baseline on {\em Beauty} and {\em Yelp2018}, but incurs a much higher computational cost.

\begin{figure}[!t]
\centering
\begin{small}
\begin{tikzpicture}
   \hspace{3mm}\begin{customlegend}[
        legend entries={\texttt{SCC},\texttt{GraphHash},\texttt{LP}},
        legend columns=3,
        area legend,
        legend style={at={(0.45,1.25)},anchor=north,draw=none,font=\footnotesize,column sep=0.25cm}]
        \addlegendimage{pattern color=white, preaction={fill, myblue}, pattern={north east lines}}  
        \addlegendimage{pattern color=white, preaction={fill, NSCcol1}, pattern={north west lines}} 
        \addlegendimage{pattern color=white, preaction={fill, cyan}, pattern={crosshatch}} 
    \end{customlegend}
    \begin{customlegend}[
        legend entries={\texttt{Leiden},\texttt{EBMD},\algo{}},
        legend columns=3,
        area legend,
        legend style={at={(0.45,0.75)},anchor=north,draw=none,font=\footnotesize,column sep=0.25cm}]   
        \addlegendimage{pattern color=white, preaction={fill, teal}, pattern={crosshatch dots}}    
        \addlegendimage{pattern color=white, preaction={fill, mypurple}, pattern={dots}}
        \addlegendimage{pattern color=white, preaction={fill, mycyan}}
    \end{customlegend}
\end{tikzpicture}
\\[-8pt]
\vspace{0mm}

\hspace{-2mm}\subfloat[{\em Beauty}]{
\begin{tikzpicture}[scale=1]
\begin{axis}[
    height=3.4cm,
    width=2.9cm,
    xtick=\empty,
    ybar=0.6pt,
    bar width=0.18cm,
    enlarge x limits=0.15,
    ylabel={\em time} (sec),
    xticklabel=\empty,
    ymin=0.1,
    ymax=200,
    ytick={0.1,1,10,100,1000},
    yticklabels={$0.1$,$1$,$10$,$10^2$,$10^3$},
    ymode=log,
    xticklabel style = {font=\scriptsize},
    yticklabel style = {font=\scriptsize},
    log origin y=infty,
    log basis y={10},
    every axis y label/.style={at={(current axis.north west)},right=5mm,above=0mm},
    legend style={draw=none, at={(1.02,1.02)},anchor=north west,cells={anchor=west},font=\scriptsize},
    legend image code/.code={ \draw [#1] (0cm,-0.1cm) rectangle (0.3cm,0.15cm);},
    ]

\addplot [pattern color=white, preaction={fill, myblue}, pattern={north east lines},nodes near coords={$\star$}, nodes near coords style={color=black, font=\normalsize, yshift=-1}] coordinates {(1,72)}; 
\addplot [pattern color=white, preaction={fill, NSCcol1}, pattern={north west lines}] coordinates {(1,0.19)}; 
\addplot [pattern color=white, preaction={fill, teal}, pattern={crosshatch dots} ] coordinates {(1,0.48)};
\addplot [pattern color=white, preaction={fill, mypurple}, pattern={dots}] coordinates {(1,0.74)};
\addplot [pattern color=white, preaction={fill, mycyan}] coordinates {(1,0.18)};

\end{axis}
\end{tikzpicture}
\label{fig:time-Beauty}%
}%
\subfloat[{\em Gowalla}]{
\begin{tikzpicture}[scale=1]
\begin{axis}[
    height=3.4cm,
    width=2.9cm,
    xtick=\empty,
    ybar=0.6pt,
    bar width=0.18cm,
    enlarge x limits=0.15,
    ylabel={\em time} (sec),
    xticklabel=\empty,
    ymin=0.1,
    ymax=1000,
    ytick={0.1,1,10,100,1000},
    yticklabels={$0.1$,$1$,$10$,$10^2$,$10^3$},
    ymode=log,
    xticklabel style = {font=\scriptsize},
    yticklabel style = {font=\scriptsize},
    log origin y=infty,
    log basis y={10},
    every axis y label/.style={at={(current axis.north west)},right=5mm,above=0mm},
    legend style={draw=none, at={(1.02,1.02)},anchor=north west,cells={anchor=west},font=\scriptsize},
    legend image code/.code={ \draw [#1] (0cm,-0.1cm) rectangle (0.3cm,0.15cm);},
    ]

\addplot [pattern color=white, preaction={fill, myblue}, pattern={north east lines}] coordinates {(1,402.3)}; 
\addplot [pattern color=white, preaction={fill, NSCcol1}, pattern={north west lines}] coordinates {(1,1.4)}; 
\addplot [pattern color=white, preaction={fill, cyan}, pattern={crosshatch}] coordinates {(1,0.2)};
\addplot [pattern color=white, preaction={fill, teal}, nodes near coords={$\star$}, pattern={crosshatch dots}] coordinates {(1,2.1)};
\addplot [pattern color=white, preaction={fill, mycyan}, nodes near coords style={color=black, font=\normalsize, yshift=-1} ] coordinates {(1,1.0)};

\end{axis}
\end{tikzpicture}
\hspace{1mm}
\label{fig:time-Gowalla}%
}%
\subfloat[{\em Yelp2018}]{
\begin{tikzpicture}[scale=1]
\begin{axis}[
    height=3.4cm,
    width=2.9cm,
    xtick=\empty,
    ybar=0.6pt,
    bar width=0.18cm,
    enlarge x limits=0.15,
    ylabel={\em time} (sec),
    xticklabel=\empty,
    ymin=0.1,
    ymax=3000,
    ytick={0.1,1,10,100,1000},
    yticklabels={$0.1$,$1$,$10$,$10^2$,$10^3$},
    ymode=log,
    xticklabel style = {font=\scriptsize},
    yticklabel style = {font=\scriptsize},
    log origin y=infty,
    log basis y={10},
    every axis y label/.style={at={(current axis.north west)},right=5mm,above=0mm},
    legend style={draw=none, at={(1.02,1.02)},anchor=north west,cells={anchor=west},font=\scriptsize},
    legend image code/.code={ \draw [#1] (0cm,-0.1cm) rectangle (0.3cm,0.15cm);},
    ]

\addplot [pattern color=white, preaction={fill, myblue}, nodes near coords={$\star$}, pattern={north east lines}] coordinates {(1,546.8)}; 
\addplot [pattern color=white, preaction={fill, NSCcol1}, pattern={north west lines}] coordinates {(1,1.54)}; 
\addplot [pattern color=white, preaction={fill, cyan}, pattern={crosshatch}] coordinates {(1,0.28)};
\addplot [pattern color=white, preaction={fill, teal}, pattern={crosshatch dots}] coordinates {(1,2.5)};
\addplot [pattern color=white, preaction={fill, mycyan}, nodes near coords style={color=black, font=\normalsize, yshift=-1} ] coordinates {(1,1.58)};

\end{axis}
\end{tikzpicture}
\hspace{1mm}
\label{fig:time-yelp}%
}%
\subfloat[{\em AmazonBook}]{
\begin{tikzpicture}[scale=1]
\begin{axis}[
    height=3.4cm,
    width=2.9cm,
    xtick=\empty,
    ybar=0.6pt,
    bar width=0.18cm,
    enlarge x limits=0.15,
    ylabel={\em time} (sec),
    xticklabel=\empty,
    ymin=1,
    ymax=3000,
    ytick={1,10,100,1000,10000},
    yticklabels={$1$,$10$,$10^2$,$10^3$,$10^4$},
    ymode=log,
    xticklabel style = {font=\scriptsize},
    yticklabel style = {font=\scriptsize},
    log origin y=infty,
    log basis y={10},
    every axis y label/.style={at={(current axis.north west)},right=5mm,above=0mm},
    legend style={draw=none, at={(1.02,1.02)},anchor=north west,cells={anchor=west},font=\scriptsize},
    legend image code/.code={ \draw [#1] (0cm,-0.1cm) rectangle (0.3cm,0.15cm);},
    ]

\addplot [pattern color=white, preaction={fill, myblue}, pattern={north east lines}] coordinates {(1,1727.2)}; 
\addplot [pattern color=white, preaction={fill, NSCcol1}, pattern={north west lines}] coordinates {(1,3.7)}; 
\addplot [pattern color=white, preaction={fill, cyan}, pattern={crosshatch}] coordinates {(1,1.85)};
\addplot [pattern color=white, preaction={fill, teal}, nodes near coords={$\star$}, pattern={crosshatch dots}] coordinates {(1,6.0)};
\addplot [pattern color=white, preaction={fill, mycyan}, nodes near coords style={color=black, font=\normalsize, yshift=-1} ] coordinates {(1,3.5)};

\end{axis}
\end{tikzpicture}
\hspace{1mm}
\label{fig:time-Amazon}%
}%
\vspace{-3mm}
\end{small}
\caption{Efficiency of strong methods in constructing sketching matrices. (best baselines in Table~\ref{tab:results} are marked with $\star$)} \label{fig:efficiency}
\vspace{-2ex}
\end{figure}
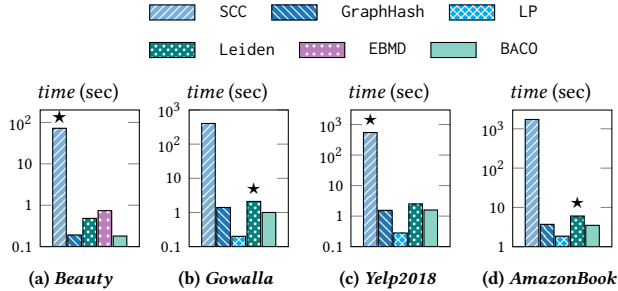

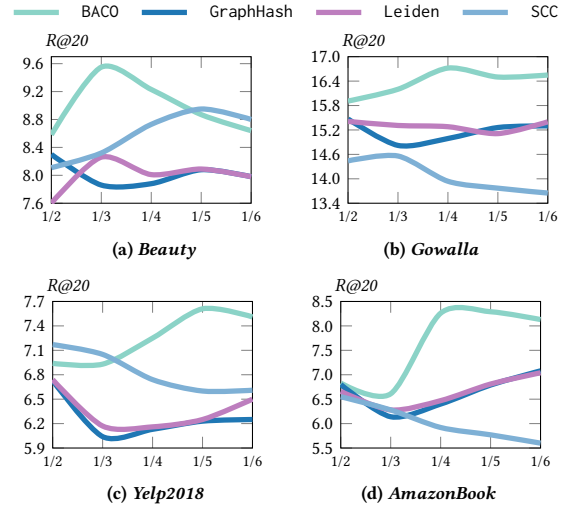
\begin{figure}[!t]
\centering
\begin{small}
\begin{tikzpicture}
    \begin{customlegend}
    [legend columns=5,
        legend entries={\algo{},\texttt{GraphHash},\texttt{Leiden},\texttt{SCC}},
        legend style={at={(0.45,1.35)},anchor=north,draw=none,font=\footnotesize,column sep=0.2cm}]
    \addlegendimage{line width=0.7mm,mark size=4pt,color=mycyan}
    \addlegendimage{line width=0.7mm,mark size=4pt,color=NSCcol1}
    \addlegendimage{line width=0.7mm,mark size=4pt,color=mypurple}
    \addlegendimage{line width=0.7mm,mark size=4pt,color=myblue}
    \end{customlegend}
\end{tikzpicture}
\\[-\lineskip]
\vspace{-4mm}
\subfloat[{\em Beauty}]{
\begin{tikzpicture}[scale=1,every mark/.append style={mark size=3pt}]
    \begin{axis}[
        height=\columnwidth/2.4,
        width=\columnwidth/2.0,
        ylabel={\it R@20},
        xmin=1, xmax=5,
        ymin=7.6, ymax=9.7,
        xtick={1.0,2.0,3.0,4.0,5.0},
        ytick={7.6,8.0,8.4,8.8,9.2,9.6},
        xticklabel style = {font=\scriptsize},
        yticklabel style = {font=\footnotesize},
        xticklabels={1/2,1/3,1/4,1/5,1/6},
        yticklabels={7.6,8.0,8.4,8.8,9.2,9.6},
        every axis y label/.style={font=\footnotesize,at={(current axis.north west)},right=2mm,above=0mm},
        legend style={fill=none,font=\small,at={(0.02,0.99)},anchor=north west,draw=none},
    ]
    \addplot[line width=0.7mm, smooth, color=mycyan]  %
        plot coordinates {
            (1,	8.58)
            (2,	9.55)
            (3,	9.23)
            (4,	8.87)
            (5,	8.64)
        };
    \addplot[line width=0.7mm, smooth, color=NSCcol1]  %
        plot coordinates {
            (1,	8.30)
            (2,	7.86)
            (3,	7.88)
            (4,	8.08)
            (5,	7.98)
        };
    \addplot[line width=0.7mm, smooth,  color=mypurple]  %
        plot coordinates {
            (1,	7.60)
            (2,	8.26)
            (3,	8.01)
            (4,	8.09)
            (5,	7.98)
        };
    \addplot[line width=0.7mm, smooth, color=myblue]  %
        plot coordinates {
            (1,	8.11)
            (2,	8.32)
            (3,	8.73)
            (4,	8.95)
            (5,	8.80)
        };

    \end{axis}
\end{tikzpicture}\hspace{4mm}\label{fig:vary-Beauty}%
}
\subfloat[{\em Gowalla}]{
\begin{tikzpicture}[scale=1,every mark/.append style={mark size=3pt}]
    \begin{axis}[
        height=\columnwidth/2.4,
        width=\columnwidth/2.0,
        ylabel={\it R@20},
        xmin=1, xmax=5,
        ymin=13.4, ymax=17.0,
        xtick={1.0,2.0,3.0,4.0,5.0},
        ytick={13.4,14.0,14.6,15.2,15.8,16.4,17.0},
        xticklabel style = {font=\scriptsize},
        yticklabel style = {font=\footnotesize},
        xticklabels={1/2,1/3,1/4,1/5,1/6},
        yticklabels={13.4,14.0,14.6,15.2,15.8,16.4,17.0},
        every axis y label/.style={font=\footnotesize,at={(current axis.north west)},right=2mm,above=0mm},
        legend style={fill=none,font=\small,at={(0.02,0.99)},anchor=north west,draw=none},
    ]
    \addplot[line width=0.7mm, smooth, color=mycyan]  %
        plot coordinates {
            (1,	15.91)
            (2,	16.20)
            (3,	16.72)
            (4,	16.50)
            (5,	16.55)
        };
    \addplot[line width=0.7mm, smooth, color=NSCcol1]  %
        plot coordinates {
            (1,	15.47)
            (2,	14.82)
            (3,	14.99)
            (4,	15.26)
            (5,	15.30)
        };
    \addplot[line width=0.7mm, smooth,  color=mypurple]  %
        plot coordinates {
            (1,	15.41)
            (2,	15.31)
            (3,	15.28)
            (4,	15.11)
            (5,	15.40)
        };
    \addplot[line width=0.7mm, smooth, color=myblue]  %
        plot coordinates {
            (1,	14.44)
            (2,	14.56)
            (3,	13.94)
            (4,	13.77)
            (5,	13.65)
        };

    \end{axis}
\end{tikzpicture}\hspace{0mm}\label{fig:vary-Gow}%
}
\vspace{-2ex}
\subfloat[{\em Yelp2018}]{
\begin{tikzpicture}[scale=1,every mark/.append style={mark size=2.2pt}]
    \begin{axis}[
        height=\columnwidth/2.4,
        width=\columnwidth/2.0,
        ylabel={\it R@20},
        xmin=1, xmax=5,
        ymin=5.9, ymax=7.7,
        xtick={1.0,2.0,3.0,4.0,5.0},
        ytick={5.9,6.2,6.5,6.8,7.1,7.4,7.7},
        xticklabel style = {font=\scriptsize},
        yticklabel style = {font=\footnotesize},
        xticklabels={1/2,1/3,1/4,1/5,1/6},
        yticklabels={5.9,6.2,6.5,6.8,7.1,7.4,7.7},
        every axis y label/.style={font=\footnotesize,at={(current axis.north west)},right=2mm,above=0mm},
        legend style={fill=none,font=\small,at={(0.02,0.99)},anchor=north west,draw=none},
    ]
    \addplot[line width=0.7mm, smooth, color=mycyan]  %
        plot coordinates {
            (1,	6.94)
            (2,	6.93)
            (3,	7.25)
            (4,	7.61)
            (5,	7.51)
        };
    \addplot[line width=0.7mm, smooth, color=NSCcol1]  %
        plot coordinates {
            (1,	6.73)
            (2,	6.04)
            (3,	6.13)
            (4,	6.23)
            (5,	6.25)
        };
    \addplot[line width=0.7mm, smooth,  color=mypurple]  %
        plot coordinates {
            (1,	6.74)
            (2,	6.17)
            (3,	6.16)
            (4,	6.25)
            (5,	6.50)
        };
    \addplot[line width=0.7mm, smooth, color=myblue]  %
        plot coordinates {
            (1,	7.17)
            (2,	7.05)
            (3,	6.74)
            (4,	6.60)
            (5,	6.61)
        };
    \end{axis}
\end{tikzpicture}\hspace{4mm}\label{fig:vary-Yelp}%
}
\subfloat[{\em AmazonBook}]{
\begin{tikzpicture}[scale=1,every mark/.append style={mark size=2.2pt}]
    \begin{axis}[
        height=\columnwidth/2.4,
        width=\columnwidth/2.0,
        ylabel={\it R@20},
        xmin=1, xmax=5,
        ymin=5.5, ymax=8.5,
        xtick={1.0,2.0,3.0,4.0,5.0},
        ytick={5.5,6,6.5,7,7.5,8,8.5},
        xticklabel style = {font=\scriptsize},
        yticklabel style = {font=\footnotesize},
        xticklabels={1/2,1/3,1/4,1/5,1/6},
        yticklabels={5.5,6.0,6.5,7.0,7.5,8.0,8.5},
        every axis y label/.style={font=\footnotesize,at={(current axis.north west)},right=2mm,above=0mm},
        legend style={fill=none,font=\small,at={(0.02,0.99)},anchor=north west,draw=none},
    ]
    \addplot[line width=0.7mm, smooth, color=mycyan]  %
        plot coordinates {
            (1,	6.83)
            (2,	6.61)
            (3,	8.26)
            (4,	8.29)
            (5,	8.13)
        };
    \addplot[line width=0.7mm, smooth, color=NSCcol1]  %
        plot coordinates {
            (1,	6.80)
            (2,	6.14)
            (3,	6.40)
            (4,	6.79)
            (5,	7.08)
        };
    \addplot[line width=0.7mm, smooth,  color=mypurple]  %
        plot coordinates {
            (1,	6.66)
            (2,	6.28)
            (3,	6.47)
            (4,	6.81)
            (5,	7.04)
        };
    \addplot[line width=0.7mm, smooth, color=myblue]  %
        plot coordinates {
            (1,	6.55)
            (2,	6.28)
            (3,	5.92)
            (4,	5.77)
            (5,	5.60)
        };
    \end{axis}
\end{tikzpicture}\hspace{0mm}\label{fig:vary-Ama}%
}
\end{small}
 \vspace{-3mm}
\caption{Performance when varying parameter ratios.} \label{fig:vary-para}
\vspace{-1ex}
\end{figure}

\subsection{Efficiency and Parameter Reduction}\label{sec:exp-efficiency}
Figure~\ref{fig:efficiency} presents the runtime costs of \algo{} and four top-performing baselines, as listed in Tables~\ref{tab:datasets-retrieval}. Note that the $y$-axis is in logarithmic scale, and the running time is measured in seconds (sec). 
Baselines with the best clustering quality are marked with $\star$. 

As shown in Figure~\ref{fig:efficiency}, \algo{} consistently offers superior efficiency while attaining the best performance across all datasets. Compared to the best baselines in Table~\ref{tab:results}, \algo{} achieves speedups of $400\times$, $2.1\times$, $346\times$, and $1.7\times$ on {\em Beauty}, {\em Gowalla}, {\em Yelp2018}, and {\em AmazonBook}, respectively. Specifically, \algo{} is at least $346\times$ faster than coclustering methods. Furthermore, when compared to highly efficient modularity maximization approaches (e.g., \texttt{GraphHash} and \texttt{Leiden}), \algo{} surpasses them both in speed and in clustering quality. While \texttt{LP} is slightly more efficient owing to its lightweight label propagation, our method consistently outperforms LP in clustering quality, achieving recall improvements in the range of $1.2\%$ to $4.9\%$.

Figure~\ref{fig:vary-para} presents the Recall@20 results under varying compression ratios ranging from $1/2$ to $1/6$. As the ratio increases, \algo{} and other graph-based algorithms exhibit distinct trends in performance across different datasets, with an overall trend of initial improvement followed by a decline. Specifically, for {\em Yelp2018} and {\em AmazonBook}, \algo{}'s performance hits a nadir at ratio $1/2$ and peaks around $1/4$ or $1/5$, while for {\em Beauty} and {\em Gowalla}, it rises initially then drops as the ratio increases. Unlike hashing methods, graph-based approaches rely on sufficiently large communities to leverage user/item homogeneity before collisions influence performance. Therefore, we recommend setting the compression ratio to at least $1/5$, as lower values can degrade performance to random hashing levels.

\begin{table}[!t]
\centering
\caption{Impact of Weighting Schemes.}
\label{tab:weighting}
\vspace{-2ex}
\begin{small}
\addtolength{\tabcolsep}{-0.4em}
\resizebox{\columnwidth}{!}{
\begin{tabular}{c|*{4}{ccc}}
\toprule
\multirow{2}{*}{\makecell{Weighting\\ Scheme}} & 
\multicolumn{2}{c}{Beauty} &
\multicolumn{2}{c}{Gowalla} & 
\multicolumn{2}{c}{Yelp2018} & 
\multicolumn{2}{c}{AmazonBook} \\
 \cmidrule(lr){2-3} \cmidrule(lr){4-5} \cmidrule(lr){6-7} \cmidrule(lr){8-9}
\multicolumn{1}{c|}{}  & \multicolumn{1}{c}{R@20\textuparrow} & \multicolumn{1}{c}{N@20\textuparrow}  & \multicolumn{1}{c}{R@20\textuparrow} & \multicolumn{1}{c}{N@20\textuparrow} & \multicolumn{1}{c}{R@20\textuparrow} & \multicolumn{1}{c}{N@20\textuparrow} & \multicolumn{1}{c}{R@20\textuparrow} & \multicolumn{1}{c}{N@20\textuparrow} \\
\midrule
 \texttt{Louvain} (HWS)
& 8.551 & 4.480
& 15.844 & 9.974
& 6.504 & 4.088
& 7.754 & 5.314 \\   

 \texttt{Louvain} (CPM)
& 8.489 & 4.417
& 16.008 & 10.364
& 7.211 & 4.588
& 6.999 & 4.747 \\   

 \texttt{Leiden} (HWS)
& 8.559 & 4.455 
& 15.705 & 10.007
& 6.563 & 4.158
& \cellcolor{blue!30}7.760 & \cellcolor{blue!30}5.314\\ 

 \texttt{Leiden} (CPM)
& \cellcolor{blue!30}8.725 & \cellcolor{blue!30}4.492
& 15.888 & 10.258
& 7.152 & 4.576
& 7.066 & 4.825 \\ \cline{1-9}

 Modularity
& 8.452 & 4.380 
& 16.246 & 10.355 
& 7.296 & 4.661
& 7.566 & 5.164 \\

 CPM
& 8.159 & 4.169 
& 16.009 & 10.262
& 6.849 & 4.341 
& 7.322 & 5.111 \\ 

 reverse HWS
& 8.159 & 4.169 
& 15.935 & 10.147
& 7.116 & 4.461
& 7.261 & 5.028\\  

\cline{1-9}

\algo{} w/o SCU
& 8.619 & 4.406 
& \cellcolor{blue!30}16.393 & \cellcolor{blue!30}10.537
& \cellcolor{blue!30}7.344 & \cellcolor{blue!30}4.679 
& 7.669 & 5.237\\
\bottomrule
\end{tabular}
}
\end{small}
\end{table}

\subsection{Impact of Weighting Schemes}\label{sec:exp-weighting}

Table~\ref{tab:weighting} presents the recall and NDCG results achieved by \algo{}, integrated with various weighting schemes as summarized in Table~\ref{tbl:framework}, within our unified framework. Overall, HWS consistently outperforms competing methods in most cases, indicating that our scheme is better suited for real-world recommendation scenarios. Specifically, when combined with HWS, \algo{} achieves average improvements of $0.12\%$ in recall and $0.07\%$ in NDCG across datasets compared to other strategies. Notably, on {\em AmazonBook}, all algorithms achieve their best performance when combined with HWS. In particular, the \texttt{Leiden} algorithm achieves the highest performance, demonstrating that HWS can be effectively integrated into various frameworks to further enhance their effectiveness. The only exception occurs on the small-scale {\em Beauty} dataset, where HWS is slightly outperformed by CPM, yet it still surpasses the baselines in Table~\ref{tab:results}.

\begin{table}[!t]
\centering
\caption{Impact of Secondary Clusters for Users.}
\label{tab:post-refine}
\vspace{-2ex}
\begin{small}
\addtolength{\tabcolsep}{-0.4em}
\resizebox{\columnwidth}{!}{
\begin{tabular}{c|*{4}{ccc}}
\toprule
\multirow{2}{*}{\makecell{Weighting\\ Scheme}} & 
\multicolumn{2}{c}{Beauty} &
\multicolumn{2}{c}{Gowalla} & 
\multicolumn{2}{c}{Yelp2018} & 
\multicolumn{2}{c}{AmazonBook} \\
 \cmidrule(lr){2-3} \cmidrule(lr){4-5} \cmidrule(lr){6-7} \cmidrule(lr){8-9}
\multicolumn{1}{c|}{}  & \multicolumn{1}{c}{R@20\textuparrow} & \multicolumn{1}{c}{N@20\textuparrow}  & \multicolumn{1}{c}{R@20\textuparrow} & \multicolumn{1}{c}{N@20\textuparrow} & \multicolumn{1}{c}{R@20\textuparrow} & \multicolumn{1}{c}{N@20\textuparrow} & \multicolumn{1}{c}{R@20\textuparrow} & \multicolumn{1}{c}{N@20\textuparrow} \\
\midrule
 \texttt{GraphHash} w/ SCU
& 9.073 & \cellcolor{blue!30}4.725
& 15.856 & 9.930
& 6.567 & 4.117
& 7.655 & 5.108 \\   

 \texttt{Leiden} w/ SCU
& 8.910 & 4.670
& 15.924 & 10.041
& 6.530 & 4.084
& 7.605 & 5.101 \\ 

 \texttt{LP} w/ SCU
& 8.428 & 4.366
& 11.517 & 7.292
& 5.120 & 3.186
& 3.540 & 2.323 \\ 

\cline{1-9}

\cline{1-9}

\algo{} w/o SCU
& 8.619 & 4.406
& 16.393 & 10.537
& 7.344 & 4.679
& 7.669 & 5.237\\ 

 \algo{} w/ SCI
& 8.530 & 4.306
& 15.892 & 10.195
& 7.209 & 4.544
& 7.393 & 5.084 \\  

 \algo{} w/ SCU \& SCI
& 8.982 & 4.555
& 16.503 & 10.476
& \cellcolor{blue!30}7.648 & \cellcolor{blue!30}4.844
& 7.801 & 5.214 \\  \cline{1-9}

\algo{}
& \cellcolor{blue!30}9.079 & 4.574
& \cellcolor{blue!30}16.692 & \cellcolor{blue!30}10.674
& 7.510 & 4.763 
& \cellcolor{blue!30}7.892 & \cellcolor{blue!30}5.240\\

\bottomrule
\end{tabular}
}
\end{small}
\end{table}

\subsection{Impact of Secondary Clusters for Users}\label{sec:exp-SCU}

 As illustrated in Table~\ref{tab:post-refine}, the integration of secondary clusters for users in \algo{} significantly reduces hash collisions, resulting in average improvements of $0.29\%$ in recall and $0.10\%$ in NDCG over the version without secondary clusters (w/o SCU). Particularly, on {\em Beauty}, \texttt{GraphHash} equipped with SCU yields notable gains of $0.9\%$ in recall and $0.5\%$ in NDCG, whereas using random secondary labels in \texttt{DoubleGraph} leads to a drop in performance. Notably, introducing secondary item clusters (w/ SCI) generally results in diminished performance, likely due to the dilution of user-level personalization as numerous items are added to clusters. In contrast, SCU sharpens user differentiation. Including both SCU and SCI decreases performance versus SCU alone, since excessive secondary labels may lead to potential conflicts. An exception occurs on {\em Yelp2018}, where \algo{} with both SCU and SCI fortuitously aligns with the dataset’s structure. Our results indicate that employing SCU alone is both sufficient and appropriate for all methods, offering a more reasonable approach than employing the double hash trick.

\section{Conclusion}
This paper presents \algo{}, a simple, fast, and effective framework to compress embedding tables through balanced co-clustering of users and items.
Through maximizing intra-cluster connectivity and enforcing balanced cluster sizes, \algo{} alleviates embedding collisions and codebook collapse, as well as aligns with established graph clustering methodologies.
We further introduce a hybrid weighting scheme, an efficient label-propagation solver that sidesteps resolution limits, and secondary user clusters that expand representational capacity.
Our experiments and ablations demonstrate competitive recommendation quality, substantial parameter reductions, and considerable speedups against strong baselines.

\begin{acks}
This work is partially supported by the National Natural Science Foundation of China (No. 62302414), the Hong Kong RGC YCRG (No. C2003-23Y), and Guangdong and Hong Kong Universities ``1+1+1'' Joint Research Collaboration Scheme, project No.: 2025A0505000002.
\end{acks}

\balance
\bibliographystyle{ACM-Reference-Format}
\bibliography{sample-base}

\pagebreak
\appendix
\section{Theoretical Proofs}

\begin{proof}[\bf Proof of Eq.~\eqref{eq:lasso}]
First, we expand ${\sum_{k=1}^K\left(\textsf{vol}(C_k)-\frac{W^{(u)}+W^{(v)}}{K}\right)^2}$:
\begin{small}
\begin{align*}
&{\sum_{k=1}^K\left[\textsf{vol}(C_k)^2+\frac{(W^{(u)}+W^{(v)})^2}{K^2}-\frac{2\textsf{vol}(C_k)\cdot (W^{(u)}+W^{(v)})}{K}\right]}.
\end{align*}
\end{small}
It can be further reorganized as
\begin{small}
\begin{align*}
& = {\sum_{k=1}^K\textsf{vol}(C_k)^2+\sum_{k=1}^K \frac{(W^{(u)}+W^{(v)})^2}{K^2}-\sum_{k=1}^K 2\textsf{vol}(C_k) \frac{(W^{(u)}+W^{(v)})}{K}} \\
& = \sum_{k=1}^K\textsf{vol}(C_k)^2+\frac{(W^{(u)}+W^{(v)})^2}{K}-\frac{2(W^{(u)}+W^{(v)})^2}{K},
\end{align*}
\end{small}
i.e., $\sum_{k=1}^K\textsf{vol}(C_k)^2-\frac{2(W^{(u)}+W^{(v)})^2}{K}$. Since $\frac{2(W^{(u)}+W^{(v)})^2}{K}$ is a constant, $\underset{C_1,\ldots,C_K}{\min}{\sum_{k=1}^K\left(\textsf{vol}(C_k)-\frac{W^{(u)}+W^{(v)}}{K}\right)^2} \Leftrightarrow \underset{C_1,\ldots,C_K}{\min}{\sum_{k=1}^K \textsf{vol}(C_k)^2}$.
By Eq.~\eqref{eq:vol}, $\sum_{k=1}^K \textsf{vol}(C_k)^2 = \sum_{k=1}^K \left(\underset{u_i\in \C_k\cap \U}{\sum}{w^{(u)}_i} + \underset{v_j\in \C_k\cap \V}{\sum}{{w^{(v)}_j}},\right)^2$ $=\|\fvec^\top\YM\|_2^2=\mathtt{Trace}(\YM^\top\fvec\fvec^\top\YM),$
which completes the proof.
\end{proof}

\begin{proof}[\bf Proof of Lemma~\ref{lem:louvain}]
By the definition of bipartite modularity by \citet{barber2007modularity} (Eq. ~\eqref{eq:bmod}), we can derive that $\sum_{k=1}^K{\left( s_k - \gamma\cdot\frac{\sigma^{(u)}_k\cdot \sigma^{(v)}_k}{|\EDG|}\right)}  = \sum_{k=1}^K{\left( \sum_{u_i, v_j\in \C_k}{\BM_{i,j}} - \gamma\cdot\frac{\sum_{u_i\in \U_k}{d(u_i)}
\cdot \sum_{v_j\in \V_k}{d(v_j)}}{|\EDG|}\right)}$,
leading to
\begin{small}
\begin{align*}
& \sum_{k=1}^K{\left( s_k - \gamma\cdot\frac{\sigma^{(u)}_k\cdot \sigma^{(v)}_k}{|\EDG|}\right)} = \sum_{k=1}^K{\sum_{u_i, v_j\in \C_k}\left( {\BM_{i,j}} - \gamma\cdot\frac{{d(u_i)}
\cdot {d(v_j)}}{|\EDG|}\right)} \\
& = \sum_{u_i\in \U} \sum_{v_j\in \V}{\left(\BM_{i,j}-\gamma\cdot \frac{d(u_i)\cdot d(v_j)}{|\EDG|}\right)\cdot \delta(u_i,v_j)}.
\end{align*}
\end{small}
Likewise, from the definition of bipartite CPM~\cite{traag2011narrow}, it follows that
\begin{small}
\begin{align*}
&\sum_{k=1}^K{\left( s_k - \gamma \cdot |\U_k|\cdot|\V_k|\right)} = \sum_{k=1}^K{\left( \sum_{u_i, v_j\in \C_k}{\BM_{i,j}} - \gamma\cdot\sum_{u_i\in \U_k}{1} \cdot \sum_{v_i\in \V_k}{1}\right)} \\
& = \sum_{k=1}^K{ \sum_{u_i, v_j\in \C_k}\left({\BM_{i,j}} - \gamma\cdot{1}\right)} =\sum_{u_i\in \U} \sum_{v_j\in \V}{\left(\BM_{i,j}-\gamma\right)\cdot \delta(u_i,v_j)}.
\end{align*}
\end{small}

As for \texttt{LPAb}~\cite{barber2009detecting}, it basically leverages label propagation to optimize the modularity by assigning a new label to $u_i$ as follows: $l(u_i) = \argmax{1\le k\le K}{\sum_{v_j\in \C_k}{\BM_{i,j} - \frac{\gamma}{\EDG} \cdot d(u_i)\cdot d(v_j) }}$. It is equivalent to $\max{\sum_{u_i\in\U}\sum_{v_j\in\V} \left( \BM_{i,j} - \gamma \cdot \frac{d(u_i)\cdot d(v_j)}{\EDG}\right)}$,
which locally maximizes the function and together sum up to an overall objective matching our final form.
\end{proof}

\begin{proof}[\bf Proof of Lemma~\ref{lem:LP}]
According to \cite{raghavan2007near,barber2009detecting}, each step in \texttt{LP} determines the cluster label of target node $u_i$ by the following way: $k^\ast = \argmax{1\le k\le K}{\sum_{v_j\in \C_k\cap \N(u_i)}{\delta(u_i,v_j)}}$,
which is to locally maximize $\sum_{v_j\in \N(u_i)}{\BM_{i,j}\cdot \delta(u_i,v_j)}$ for node $u_i$. If we consider all nodes in $\U$, it leads to an overall objective of \texttt{LP}:
\begin{small}
\begin{align*}
\max_{\YM}\sum_{u_i\in \U} \sum_{v_j\in \N(u_i)}{\BM_{i,j}\cdot \delta(u_i,v_j)} = \max_{\YM}\sum_{u_i\in \U}\sum_{v_j\in \V}{\BM_{i,j}\cdot \delta(u_i,v_j)}.
\end{align*}
\end{small}

Since \texttt{SCC} is the version of spectral clustering on bipartite graphs, whose objective function is often framed as a trace minimization problem: $\min_{\YM}\textsf{Trace}(\YM^\top(\DM-\AM)\YM) \Leftrightarrow \min_{\YM}\textsf{Trace}(\YM^\top\DM\YM) - \textsf{Trace}(\YM^\top\AM\YM)$, 
which can be simplified as $\max_{\YM}\textsf{Trace}(\YM^\top\AM\YM)$ since $\textsf{Trace}(\YM^\top\DM\YM)=\sum_{k=1}^K\sum_{x\in \C_k}{d(x)}=|\EDG|$ is a constant. This objective can also be rewritten as $\max_{\YM}{\sum_{u_i\in \U} \sum_{v_j\in \N(u_i)}{\BM_{i,j}\cdot \delta(u_i,v_j)}}$ by Eq.~\eqref{eq:trace-to-sum}.
\end{proof}

\section{Detailed Related Works}\label{sec:add-rw}
\subsection{Co-Clustering}
Co-clustering seeks to reveal associations between row and column clusters through the simultaneous clustering of both dimensions in a data matrix. Initially,~\citet{hartigan1972direct} advocated co-clustering of variables and cases, enabling direct interpretation of the resulting clusters.~\citet{dhillon2001co} established a connection between co-clustering and spectral graph partitioning. Regarding the spectral co-clustering algorithm,~\citet{kluger2003spectral} enabled different cluster numbers per dimension,~\citet{gao2005consistent} resolved higher-order problems, and~\citet{shi2010efficient} integrated constraint information. Employing information theory transforms co-clustering into an optimization problem, where different association measures yield distinct algorithms: ~\cite{govaert1995simultaneous,hartigan1972direct} utilizes a least-squares criterion, \texttt{ITCC}~\cite{dhillon2003information} maximizes mutual information, and \texttt{BCC}~\cite{banerjee2004generalized} optimizes the Bregman divergence. Numerous model-based and matrix factorization-based algorithms have also been proposed, as reviewed in detail in previous surveys~\cite{Wang2023asurvey,battaglia2024co}. Recently, co-clustering techniques have accelerated improvements in recommender systems. For instance,~\citet{wu2016explaining} incorporated sampling techniques to accelerate recommendations, while~\citet{feng2020improving} improved accuracy by limiting recommended items to the same cluster. However, these methods are typically employed to supplement information and enhance accuracy, with limited focus on addressing the memory constraints of embedding tables.

\subsection{Bipartite Graph Clustering}
Bipartite graph clustering is a method for uncovering underlying structural properties in diverse relationship networks~\cite{yang2024effective}. A straightforward approach is to transform the bipartite graph into a unipartite graph, thereby allowing the application of conventional graph clustering methods, e.g.,~\cite{newman2004finding,von2007tutorial}. Moreover, projection-based methods~\cite{melamed2014community,tackx2017comsim,yang2024efficient} generate a unipartite graph to yield higher-quality clusters, but this often results in a much denser graph structure. Approaches specifically developed for bipartite graph clustering include spectral clustering~\cite{dhillon2001co,kluger2003spectral}, statistical modeling~\cite{larremore2014efficiently,yen2020community}, and graph embeddings~\cite{yang2022scalable}. However, the aforementioned methods often involve considerable computational overhead, whereas modularity maximization~\cite{blondel2008fast,traag2019louvain,kim2022abc} and label propagation~\cite{raghavan2007near,taguchi2020bimlpa} methods are recognized for their computational efficiency. Its influence diffusion-based structure enables continuous and enhanced information exchange between users and items, making it particularly well-suited for recommendation systems. Recently,~\citet{wu2025graphhash} innovatively exploited user-item interaction graphs to compress embedding tables for recommendation tasks. However, their approach simply applies modularity maximization, without adequately considering the unique data characteristics and clustering biases of recommender systems.

\section{Additional Experimental Details}\label{sec:add-exp-details}

\subsection{Datasets Details}
We conduct our experiments on four benchmark datasets, each widely utilized in recommendation research~\cite{wu2025graphhash,he2020lightgcn,wang2022towards} and real-world scenarios. The datasets are detailed as follows:
\begin{itemize}[leftmargin=*]
\item {\em Beauty}: A subset of Amazon product reviews, encompassing user interactions of beauty products.
\item {\em Gowalla}: A check-in dataset capturing user location-sharing behaviors on the Gowalla platform.
\item {\em Yelp2018}: Extracted from the 2018 Yelp Challenge, this dataset contains user interactions with local businesses.
\item {\em AmazonBook}: A subset of Amazon product reviews, containing user interactions with books.
\end{itemize}

\subsection{Parameter Settings}\label{sec:param-set}
\begin{table}[!ht]
\small
\centering
\caption{Parameter setting in \algo{}}
\label{tab:param-gamma-BACO}
\vspace{-3ex}
\setlength{\tabcolsep}{3.5pt}
\begin{tabular}{@{}c*{6}{c}@{}}
\toprule
\cmidrule(l{3pt}r{3pt}){2-5}
Parameter & {\em Beauty} & {\em Gowalla} & {\em Yelp2018} & {\em AmazonBook}\\
\midrule
$\gamma$ & 0.13 & 7.57 & 5.50 & 4.73 \\
\bottomrule
\end{tabular}
\vspace{-2ex}
\end{table}
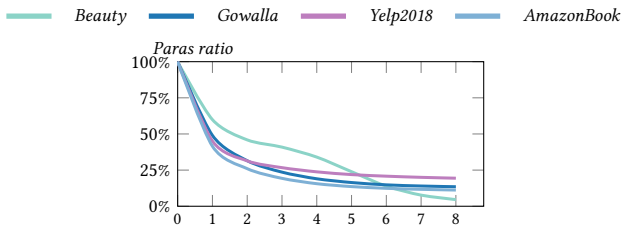
\begin{figure}[!ht]
\centering
\begin{small}
\begin{tikzpicture}
    \begin{customlegend}
    [legend columns=4,
        legend entries={{\em Beauty},{\em Gowalla}, {\em Yelp2018}, {\em AmazonBook}},
        legend style={at={(0.45,1.35)},anchor=north,draw=none,font=\footnotesize,column sep=0.2cm}]
    \addlegendimage{line width=0.7mm,color=mycyan}
    \addlegendimage{line width=0.7mm,color=NSCcol1}
    \addlegendimage{line width=0.7mm,color=mypurple}
    \addlegendimage{line width=0.7mm,color=myblue}
    \end{customlegend}
\end{tikzpicture}
\\[-\lineskip]
\begin{tikzpicture}[scale=1,every mark/.append style={mark size=1.3pt}]
    \begin{axis}[
        height=\columnwidth/2.418,
        width=\columnwidth/1.5,
        ylabel={\it Paras ratio},
        xmin=0,
        ymin=0, ymax=100,
        xtick={0,1,2,3,4,5,6,7,8},
        ytick={0,25,50,75,100},
        xticklabel style = {font=\scriptsize},
        yticklabel style = {font=\footnotesize},
        yticklabels={$0\%$,$25\%$,$50\%$,$75\%$,$100\%$},
        every axis y label/.style={font=\footnotesize,at={(current axis.north west)},right=2mm,above=0mm},
        legend style={fill=none,font=\small,at={(0.02,0.99)},anchor=north west,draw=none},
    ]
    \addplot[line width=0.4mm, smooth, color=mycyan]  %
        plot coordinates {
            (0, 100)
            (1, 59.9)
            (2, 45.9)
            (3, 40.9)
            (4, 34.0)
            (5, 23.9)
            (6, 13.9)
            (7, 7.7)
            (8, 4.6)

        };
    \addplot[line width=0.4mm, smooth, color=NSCcol1]  %
        plot coordinates {
            (0, 100)
            (1, 49.0)
            (2, 31.5)
            (3, 23.6)
            (4, 19.0)
            (5, 16.4)
            (6, 14.8)
            (7, 14.0)
            (8, 13.5)

        };
    \addplot[line width=0.4mm, smooth, color=mypurple]  %
        plot coordinates {
            (0, 100)
            (1, 45.1)
            (2, 31.4)
            (3, 26.7)
            (4, 23.8)
            (5, 21.9)
            (6, 20.8)
            (7, 20.0)
            (8, 19.4)
        };
    \addplot[line width=0.4mm, smooth, color=myblue]  %
        plot coordinates {
            (0, 100)
            (1, 41.5)
            (2, 26.0)
            (3, 19.3)
            (4, 15.6)
            (5, 13.6)
            (6, 12.4)
            (7, 11.7)
            (8, 11.2)
        };
    \end{axis}
\end{tikzpicture}\hspace{2mm}%
\end{small}
 \vspace{-3mm}
\caption{Embedding table parameters ratio of \algo{} versus iteration count.} \label{fig:convergence}
\vspace{-1ex}
\end{figure}
In this section, we present the parameters not detailed in the main text. We utilize the Adam~\cite{kingma2014adam} optimizer with a learning rate of 0.001 and a mini-batch size of 1024, and an embedding dimension of 64 across all datasets. Training is conducted for up to 1000 epochs, with early stopping(patience of 50 epochs) and validation strategies employed to prevent overfitting.

Following the settings in~\cite{wu2025graphhash}, we assign half of the hash bins to the highest-frequency entities in both the \texttt{Frequency Hashing}~\cite{ghaemmaghami2022learning} and \texttt{Hybrid Hashing}, and employ the user-item interaction graph as the feature in \texttt{LSH}~\cite{datar2004locality}.
In addition, \texttt{CCE}~\cite{tsang2023clustering} and \texttt{LEGCF}~\cite{liang2024lightweight} require dynamic updates, but for fairness, their sketching matrices are updated only in the first epoch, while hash labels for other methods are fixed before training.
We retain the original implementations of \texttt{infomap}~\cite{rosvall2008maps}, \texttt{BiMLPA}~\cite{taguchi2020bimlpa}, \texttt{BRIM}~\cite{platig2016bipartite}, as these methods inherently perform adaptive cluster detection without explicit control over the number of communities.

For a fair comparison, we adjust our parameter $\gamma$ to match the target size of the embedding table,
as in Table~\ref{tab:param-gamma-BACO}. As shown in Figure~\ref{fig:convergence}, our empirical study confirms that the parameter converges at an exponential rate. The parameter achieves around 20\% compression ratio and tends to be stable after $5$ iterations. Thereon, we fix the parameter $T$ to $5$.

\subsection{Evaluation Metrics}
Given a set of $\C=\{\C_1,\C_2,\ldots,\C_K\}$ of $K$ disjoint co-clusters, each containing both users and items, we provide the formal mathematical definitions of the {\em averaged cross-cluster links}(ACCL) and the {\em Gini coefficient} in Figure~\ref{fig:preliminary} as follows:
\begin{small}
\begin{equation*}
\textnormal{ACCL}=\frac{\underset{\C_k, \C_\ell \in \C}{\sum}\underset{\ u_i\in \C_k, v_j \in \C_\ell}{\sum}{\BM_{i,j}}}{\binom{K}{2}},\ 
\textnormal{Gini}=\frac{2}{K}\cdot \sum_{i=1}^{K} \left( \frac{i}{K} - \frac{\sum_{j=1}^{i}|\C_j|}{\sum_{k=1}^{K}|\C_k|} \right).
\end{equation*}  
\end{small}

\section{Additional Experimental Results}

\begin{table*}[!t]
\centering
\caption{Recommendation Performance ($k=10$ or $50$). Best results highlighted in {\color{blue!30}blue} and runner-up \underline{underlined}.}
\label{tab:more-result}
\vspace{-2ex}
\begin{small}
\renewcommand{\arraystretch}{0.9}
\addtolength{\tabcolsep}{-0.4em}
\resizebox{\textwidth}{!}{
\begin{tabular}{c|*{4}{cccc}}
\toprule
\multirow{2}{*}{Method} & 
\multicolumn{4}{c}{Beauty} & 
\multicolumn{4}{c}{Gowalla} & 
\multicolumn{4}{c}{Yelp2018} & 
\multicolumn{4}{c}{AmazonBook} \\
 \cmidrule(lr){2-5} \cmidrule(lr){6-9} \cmidrule(lr){10-13} \cmidrule(lr){14-17}
\multicolumn{1}{c|}{} & 
\multicolumn{1}{c}{R@10\textuparrow} & \multicolumn{1}{c}{N@10\textuparrow} &  \multicolumn{1}{c}{R@50\textuparrow} & \multicolumn{1}{c}{N@50\textuparrow} & \multicolumn{1}{c}{R@10\textuparrow} & \multicolumn{1}{c}{N@10\textuparrow} & 
\multicolumn{1}{c}{R@50\textuparrow} & \multicolumn{1}{c}{N@50\textuparrow} & \multicolumn{1}{c}{R@10\textuparrow} & \multicolumn{1}{c}{N@10\textuparrow} & \multicolumn{1}{c}{R@50\textuparrow} & \multicolumn{1}{c}{N@50\textuparrow} & \multicolumn{1}{c}{R@10\textuparrow} & \multicolumn{1}{c}{N@10\textuparrow} & \multicolumn{1}{c}{R@50\textuparrow} & \multicolumn{1}{c}{N@50\textuparrow} \\
\midrule

 Full Model
& 7.619 & 4.814 & 16.343 & 6.998
& 12.917 & 9.837 & 29.010 & 14.538
& 5.602 & 4.334 & 16.334 & 7.796
& 5.689 & 4.542 & 15.093 & 7.603
\\ \cline{1-17}

 \texttt{Random}
& 3.054 & 1.874 & 7.555 & 3.009 
& 6.368 & 5.015 & 15.097 & 7.588
& 3.048 & 2.414 & 9.055 & 4.375
& 1.677 & 1.445 & 4.679 & 2.448\\

 \texttt{Frequency Hashing}~\cite{ghaemmaghami2022learning} 
& 2.926 & 1.763 & 7.637 & 2.945
& 5.852 & 4.597 & 14.331 & 7.108
& 3.071 & 2.451 & 9.262 & 4.482
& 1.459 & 1.243 & 4.095 & 2.127\\

 \texttt{Double Hashing}~\cite{zhang2020model}
& 2.872 & 1.717 & 7.035 & 2.779
& 6.491 & 5.133 & 15.172 & 7.705
& 3.170 & 2.504 & 9.542 & 4.587
& 1.626 & 1.374 & 4.393 & 2.304\\

 \texttt{Hybrid Hashing}~\cite{zhang2020model}
& 3.028 & 1.786 & 7.704 & 2.976
& 6.901 & 5.451 & 16.061 & 8.151
& 3.856 & 3.112 & 11.172 & 5.486
& 1.726 & 1.418 & 4.892 & 2.471\\

 \texttt{LSH}~\cite{datar2004locality}
& 3.613 & 2.179 & 8.720 & 3.470
& 6.481 & 5.191 & 14.244 & 7.511
& 2.858 & 2.255 & 8.274 & 4.033
& 2.287 & 1.899 & 6.260 & 3.212\\

\texttt{CCE}~\cite{tsang2023clustering}
& 4.953 & 2.993 & 11.478 & 4.629
& 7.391 & 5.844 & 16.710 & 8.586
& 3.533 & 2.808 & 10.148 & 4.960
& 2.117 & 1.710 & 5.937 & 2.975\\

 \texttt{LEGCF}~\cite{liang2024lightweight}
& 4.989 & 2.914 & 12.391 & 4.754
& 5.413 & 4.046 & 13.230 & 6.319
& 2.747 & 2.071 & 8.060 & 3.795
& - & -  & - & - \\ \cline{1-17}

 \texttt{GraphHash}~\cite{wu2025graphhash}
& 5.776 & 3.579 & 12.294 & 5.224
& 10.672 & 8.074 & \underline{24.442} & 12.112
& 3.863 & 3.013 & 11.709 & 5.542
& 5.081 & 4.264 & 11.851 & 6.477\\

 \texttt{DoubleGraphHash}~\cite{wu2025graphhash}
& 3.883 & 2.439 & 9.307 & 3.826
& 8.711 & 6.720 & 19.949 & 10.007
& 3.245 & 2.529 & 9.708 & 4.635
& 3.708 & 2.778 & 9.221 & 4.519\\

 \texttt{LP}~\cite{raghavan2007near}
& 5.465 & 3.415 & 12.162 & 5.094
& 8.026 & 6.253 & 19.144 & 9.481
& 3.358 & 2.636 & 10.291 & 4.868
& 2.278 & 1.830 & 6.435 & 3.181\\

 \texttt{Leiden}~\cite{traag2019louvain}
& 5.540 & 3.446 & 12.165 & 5.108
& \underline{10.753} & \underline{8.203} & 24.245 & \underline{12.152}
& 3.898 & 3.000 & 11.764 & 5.536
& \underline{5.174} & \underline{4.300} & \underline{11.927} & \underline{6.513}\\

 \texttt{EBMD}~\cite{kim2022abc}
& 5.443 & 3.290 & 11.961 & 4.949
& 7.061 & 5.501 & 16.816 & 8.352
& 3.080 & 2.407 & 9.628 & 4.525
& 2.182 & 1.780 & 6.052 & 3.058 \\

 \texttt{infomap}~\cite{rosvall2008maps}
& 2.922 & 1.767 & 7.521 & 2.893
& 4.085 & 2.997 & 9.511 & 4.556
& 2.690 & 2.081 & 8.074 & 3.827
& 0.749 & 0.584 & 2.248 & 1.067\\

 \texttt{BiMLPA}~\cite{taguchi2020bimlpa} 
& 2.293 & 1.295 & 6.806 & 2.381
& 4.423 & 3.262 & 10.430 & 4.986
& 2.751 & 2.140 & 8.284 & 3.930
& 0.699 & 0.545 & 2.106 & 0.996\\

  \texttt{BRIM}~\cite{platig2016bipartite}
& 5.418 & 3.236 & 12.875 & 5.082
& 5.595 & 4.378 & 12.935 & 6.518
& 2.702 & 2.080 & 8.025 & 3.809
& 1.438 & 1.121 & 4.387 & 2.081\\ \cline{1-17}

 \texttt{SCC}~\cite{dhillon2001co} 
& \underline{5.920} & \underline{3.640} & \underline{13.072} & \underline{5.456}
& 9.667 & 7.467 & 21.880 & 11.016
& \underline{4.167} & \underline{3.224} & \underline{12.246} & \underline{5.832}
& 3.622 & 2.909 & 9.869 & 4.941 \\

 \texttt{SBC}~\cite{kluger2003spectral} 
& 4.360 & 2.630 & 10.506 & 4.192
& 7.220 & 5.665 & 17.084 & 8.522
& 3.060 & 2.377 & 9.164 & 4.354
& 2.389 & 1.944 & 6.508 & 3.300\\

 \texttt{ITCC}~\cite{dhillon2003information}
& 3.096 & 1.925 & 7.648 & 3.090
& 5.461 & 4.419 & 12.163 & 6.439
& 2.788 & 2.258 & 8.291 & 4.069
& 1.753 & 1.471 & 4.701 & 2.465 \\ \cline{1-17}

\algo{}
& \cellcolor{blue!30}6.308 & \cellcolor{blue!30}3.733 & \cellcolor{blue!30}14.036 & \cellcolor{blue!30}5.643
& \cellcolor{blue!30}11.701 & \cellcolor{blue!30}9.052 & \cellcolor{blue!30}26.429 & \cellcolor{blue!30}13.344
& \cellcolor{blue!30}4.790 & \cellcolor{blue!30}3.767 & \cellcolor{blue!30}13.739 & \cellcolor{blue!30}6.668
& \cellcolor{blue!30}5.357 & \cellcolor{blue!30}4.315 & \cellcolor{blue!30}13.221 & \cellcolor{blue!30}6.868 \\

 v.s. Baselines
& +0.388 & +0.093 & +0.964 & +0.187
& +0.948 & +0.849 & +1.987 & +1.192
& +0.623 & +0.543 & +1.493 & +0.836
& +0.183 & +0.015 & +1.294 & +0.355 \\

 v.s. Full Model
& -1.311 & -1.081 & -2.307 & -1.355
& -1.216 & -0.785 & -2.581 & -1.194
& -0.812 & -0.568 & -2.596 & -1.127
& -0.332 & -0.227 & -1.872 & -0.734\\

\bottomrule
\end{tabular}
}
\end{small}
\end{table*}

\subsection{Recommendation Performance Evaluation}
To comprehensively assess the performance of \algo{}, we additionally report Recall@$K$ and NDCG@$K$ for $K=10$ and $K=50$, using the same embedding table size as shown in Table~\ref{tab:results}. As shown in Table~\ref{tab:more-result}, \algo{} consistently outperforms all baseline approaches across all datasets, with improvements of up to $1.987\%$ in Recall and $1.192\%$ in NDCG, achieving substantial improvements in both evaluation settings. Overall, these results demonstrate the substantial effectiveness of \algo{} under both strict and relaxed scenarios.

\begin{figure}[!t]
\centering
\begin{small}
\begin{tikzpicture}
   \hspace{1mm}\begin{customlegend}[
        legend entries={\texttt{GraphHash},\texttt{Leiden},\texttt{BACO},\texttt{full}},
        legend columns=5,
        area legend,
        legend style={at={(0.45,1.25)},anchor=north,draw=none,font=\footnotesize,column sep=0.2cm}]
        \addlegendimage{pattern color=white, preaction={fill, NSCcol1}, pattern={north west lines}} 
        \addlegendimage{pattern color=white, preaction={fill, teal}, pattern={crosshatch dots}}
        \addlegendimage{pattern color=white, preaction={fill, mycyan}} 
        \addlegendimage{pattern color=white, preaction={fill, myblue}, pattern={crosshatch}}   
    \end{customlegend}
\end{tikzpicture}
\\[-8pt]
\vspace{-1ex}
\subfloat[{\em Gowalla (Recall)}]{
\begin{tikzpicture}[scale=1]
\begin{axis}[
    height=3cm,
    width=4.6cm,
    xtick=\empty,
    ybar=0.6pt,
    bar width=0.16cm,
    enlarge x limits=0.15,
    ylabel={\em R@20},
    ymin=12,
    ymax=22,
    ytick={12,14,16,18,20,22},
    yticklabels={$12$,$14$,$16$,$18$,$20$,$22$},
    xticklabel style = {font=\scriptsize},
    yticklabel style = {font=\scriptsize},
    symbolic x coords={0-25\%,25-50\%,50-75\%,75-100\%},
    xtick = data,
    tick align=inside,
    group style={group size=5 by 1, horizontal sep=0.7cm},
    every axis y label/.style={at={(current axis.north west)},right=5mm,above=0mm},
    legend style={draw=none, at={(1.02,1.02)},anchor=north west,cells={anchor=west},font=\scriptsize},
    legend image code/.code={ \draw [#1] (0cm,-0.1cm) rectangle (0.3cm,0.15cm);},
    ]

\addplot [pattern color=white, preaction={fill, NSCcol1}, pattern={north west lines}] coordinates
    {(0-25\%,16.97) (25-50\%,16.74) (50-75\%,15.60) (75-100\%,12.73)};
\addplot [pattern color=white, preaction={fill, teal}, pattern={crosshatch dots}] coordinates
    {(0-25\%,17.35) (25-50\%,16.34) (50-75\%,15.54) (75-100\%,12.88)};
\addplot [pattern color=white, preaction={fill, mycyan}] coordinates
    {(0-25\%,18.91) (25-50\%,17.78) (50-75\%,16.89) (75-100\%,13.90)};
\addplot [pattern color=white, preaction={fill, myblue}, pattern={crosshatch}] coordinates
    {(0-25\%,21.01) (25-50\%,19.50) (50-75\%,18.35) (75-100\%,14.64)};

\end{axis}
\end{tikzpicture}\hspace{1mm}\label{fig:group-recall}%
}%
\subfloat[{\em Gowalla (NDCG)}]{
\begin{tikzpicture}[scale=1]
\begin{axis}[
    height=3cm,
    width=4.6cm,
    xtick=\empty,
    ybar=0.6pt,
    bar width=0.16cm,
    enlarge x limits=0.15,
    ylabel={\em N@20} ,
    ymin=9,
    ymax=12,
    ytick={9,9.5,10,10.5,11,11.5,12},
    yticklabels={$9.0$,$9.5$,$10.0$,$10.5$,$11.0$,$11.5$,$12.0$},
    xticklabel style = {font=\scriptsize},
    yticklabel style = {font=\scriptsize},
    symbolic x coords={0-25\%,25-50\%,50-75\%,75-100\%},
    xtick = data,
    tick align=inside,
    group style={group size=5 by 1, horizontal sep=0.7cm},
    every axis y label/.style={at={(current axis.north west)},right=5mm,above=0mm},
    legend style={draw=none, at={(1.02,1.02)},anchor=north west,cells={anchor=west},font=\scriptsize},
    legend image code/.code={ \draw [#1] (0cm,-0.1cm) rectangle (0.3cm,0.15cm);},
    ]

\addplot [pattern color=white, preaction={fill, NSCcol1}, pattern={north west lines}] coordinates
    {(0-25\%,9.35) (25-50\%,9.16) (50-75\%,9.54) (75-100\%,10.28)};
\addplot [pattern color=white, preaction={fill, teal}, pattern={crosshatch dots}] coordinates
    {(0-25\%,9.53) (25-50\%,9.1) (50-75\%,9.69) (75-100\%,10.46)};
\addplot [pattern color=white, preaction={fill, mycyan}] coordinates
    {(0-25\%,10.37) (25-50\%,10.05) (50-75\%,10.71) (75-100\%,11.43)};
\addplot [pattern color=white, preaction={fill, myblue}, pattern={crosshatch}] coordinates
    {(0-25\%,11.74) (25-50\%,11.01) (50-75\%,11.64) (75-100\%,11.84)};

\end{axis}
\end{tikzpicture}\hspace{1mm}\label{fig:group-ndcg}%
}%
\vspace{-1em}

\subfloat[{\em Yelp2018 (Recall)}]{
\begin{tikzpicture}[scale=1]
\begin{axis}[
    height=3cm,
    width=4.6cm,
    xtick=\empty,
    ybar=0.6pt,
    bar width=0.16cm,
    enlarge x limits=0.15,
    ylabel={\em R@20},
    ymin=5,
    ymax=10,
    ytick={5,6,7,8,9,10},
    yticklabels={5,6,7,8,9,10},
    xticklabel style = {font=\scriptsize},
    yticklabel style = {font=\scriptsize},
    symbolic x coords={0-25\%,25-50\%,50-75\%,75-100\%},
    xtick = data,
    tick align=inside,
    group style={group size=5 by 1, horizontal sep=0.7cm},
    every axis y label/.style={at={(current axis.north west)},right=5mm,above=0mm},
    legend style={draw=none, at={(1.02,1.02)},anchor=north west,cells={anchor=west},font=\scriptsize},
    legend image code/.code={ \draw [#1] (0cm,-0.1cm) rectangle (0.3cm,0.15cm);},
    ]

\addplot [pattern color=white, preaction={fill, NSCcol1}, pattern={north west lines}] coordinates
    {(0-25\%,6.48) (25-50\%,6.93) (50-75\%,6.72) (75-100\%,6.42)};
\addplot [pattern color=white, preaction={fill, teal}, pattern={crosshatch dots}] coordinates
    {(0-25\%,6.05) (25-50\%,6.21) (50-75\%,6.53) (75-100\%,6.15)};
\addplot [pattern color=white, preaction={fill, mycyan}] coordinates
    {(0-25\%,7.51) (25-50\%,7.91) (50-75\%,7.47) (75-100\%,7.31)};
\addplot [pattern color=white, preaction={fill, myblue}, pattern={crosshatch}] coordinates
    {(0-25\%,8.83) (25-50\%,9.63) (50-75\%,8.93) (75-100\%,8.11)};

\end{axis}
\end{tikzpicture}\hspace{1mm}\label{fig:more-group-recall}%
}%
\subfloat[{\em Yelp2018 (NDCG)}]{
\begin{tikzpicture}[scale=1]
\begin{axis}[
    height=3cm,
    width=4.6cm,
    xtick=\empty,
    ybar=0.6pt,
    bar width=0.16cm,
    enlarge x limits=0.15,
    ylabel={\em N@20} ,
    ymin=3,
    ymax=7,
    ytick={3,4,5,6,7},
    yticklabels={$9.0$,$9.5$,$10.0$,$10.5$,$11.0$,$11.5$,$12.0$},
    xticklabel style = {font=\scriptsize},
    yticklabel style = {font=\scriptsize},
    symbolic x coords={0-25\%,25-50\%,50-75\%,75-100\%},
    xtick = data,
    tick align=inside,
    group style={group size=5 by 1, horizontal sep=0.7cm},
    every axis y label/.style={at={(current axis.north west)},right=5mm,above=0mm},
    legend style={draw=none, at={(1.02,1.02)},anchor=north west,cells={anchor=west},font=\scriptsize},
    legend image code/.code={ \draw [#1] (0cm,-0.1cm) rectangle (0.3cm,0.15cm);},
    ]

\addplot [pattern color=white, preaction={fill, NSCcol1}, pattern={north west lines}] coordinates
    {(0-25\%,3.40) (25-50\%,3.73) (50-75\%,4.00) (75-100\%,5.34)};
\addplot [pattern color=white, preaction={fill, teal}, pattern={crosshatch dots}] coordinates
    {(0-25\%,3.18) (25-50\%,3.31) (50-75\%,3.83) (75-100\%,5.07)};
\addplot [pattern color=white, preaction={fill, mycyan}] coordinates
    {(0-25\%,4.10) (25-50\%,4.35) (50-75\%,4.54) (75-100\%,6.07)};
\addplot [pattern color=white, preaction={fill, myblue}, pattern={crosshatch}] coordinates
    {(0-25\%,4.74) (25-50\%,5.30) (50-75\%,5.44) (75-100\%,6.61)};
\end{axis}
\end{tikzpicture}\hspace{1mm}\label{fig:more-group-ndcg}%
}%
\vspace{-1em}
\end{small}
\caption{Performance breakdown by test user frequency.} \label{fig:subgroup}
\vspace{-2ex}
\end{figure}
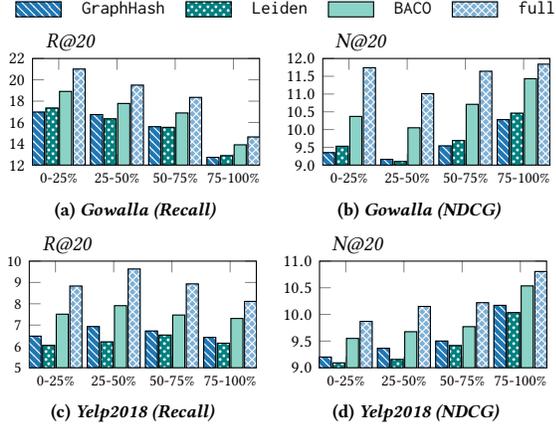
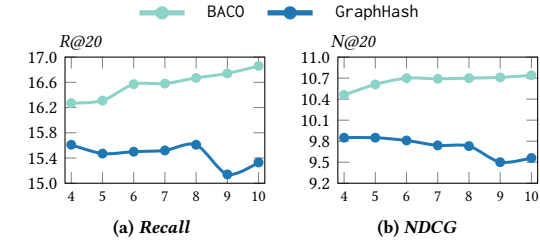
\begin{figure}[!t]
\centering
\begin{small}
\begin{tikzpicture}
    \begin{customlegend}
    [legend columns=5,
        legend entries={\algo{},\texttt{GraphHash}},
        legend style={at={(0.45,1.35)},anchor=north,draw=none,font=\footnotesize,column sep=0.2cm}]
    \addlegendimage{line width=0.7mm,mark=*,mark size=1.5pt,color=mycyan}
    \addlegendimage{line width=0.7mm,mark=*, mark size=1.5pt,color=NSCcol1}
    \end{customlegend}
\end{tikzpicture}
\\[-\lineskip]
\vspace{-4mm}
\subfloat[{\em Recall}]{
\begin{tikzpicture}[scale=1,every mark/.append style={mark size=1.3pt}]
    \begin{axis}[
        height=\columnwidth/2.6,
        width=\columnwidth/2.0,
        ylabel={\it R@20},
        xmin=1.8, xmax=8.2,
        ymin=15, ymax=17,
        xtick={2,3,4,5,6,7,8},
        ytick={15.0,15.4,15.8,16.2,16.6,17.0},
        xticklabel style = {font=\scriptsize},
        yticklabel style = {font=\footnotesize},
        xticklabels={4,5,6,7,8,9,10},
        yticklabels={15.0,15.4,15.8,16.2,16.6,17.0},
        every axis y label/.style={font=\footnotesize,at={(current axis.north west)},right=2mm,above=0mm},
        legend style={fill=none,font=\small,at={(0.02,0.99)},anchor=north west,draw=none},
    ]
    \addplot[line width=0.4mm, smooth, mark=*, color=mycyan]  %
        plot coordinates {
            (2,	16.27)
            (3,	16.31)
            (4,	16.57)
            (5,	16.58)
            (6,	16.67)
            (7, 16.74)
            (8,	16.86)
        };
    \addplot[line width=0.4mm, smooth, mark=*, color=NSCcol1]  %
        plot coordinates {
            (2,	15.61)
            (3,	15.47)
            (4,	15.50)
            (5, 15.52)
            (6,	15.61)
            (7,	15.14)
            (8,	15.33)
        };

    \end{axis}
\end{tikzpicture}\hspace{2mm}\label{fig:gamma-R}%
}
\subfloat[{\em NDCG}]{
\begin{tikzpicture}[scale=1,every mark/.append style={mark size=1.3pt}]
    \begin{axis}[
        height=\columnwidth/2.6,
        width=\columnwidth/2.0,
        ylabel={\it N@20},
        xmin=1.8, xmax=8.2,
        ymin=9.2, ymax=11,
        xtick={2,3,4,5,6,7,8},
        ytick={9.2, 9.5,9.8,10.1,10.4,10.7,11.0},
        xticklabel style = {font=\scriptsize},
        yticklabel style = {font=\footnotesize},
        xticklabels={4,5,6,7,8,9,10},
        yticklabels={9.2, 9.5,9.8,10.1,10.4,10.7,11.0},
        every axis y label/.style={font=\footnotesize,at={(current axis.north west)},right=2mm,above=0mm},
        legend style={fill=none,font=\small,at={(0.02,0.99)},anchor=north west,draw=none},
    ]
    \addplot[line width=0.4mm, smooth, mark=*, color=mycyan]  %
        plot coordinates {
            (2,	10.46)
            (3,	10.61)
            (4,	10.70)
            (5,	10.69)
            (6,	10.70)
            (7, 10.71)
            (8,	10.74)
        };
    \addplot[line width=0.4mm, smooth, mark=*, color=NSCcol1]  %
        plot coordinates {
            (2,	9.85)
            (3,	9.85)
            (4,	9.81)
            (5,	9.74)
            (6,	9.73)
            (7, 9.50)
            (8,	9.56)
        };

    \end{axis}
\end{tikzpicture}\hspace{0mm}\label{fig:vary-Gow}%
}

\end{small}
 \vspace{-3mm}
\caption{Impact of resolution paramater $\gamma$.} \label{fig:vary-gamma}
\vspace{-2ex}
\end{figure}

\subsection{User Subgroup Evaluation}
This experiment investigates the efficacy of algorithms across user groups, which are categorized by their activity frequency percentiles in the training data. In the Figure~\ref{fig:subgroup}, we report the average metrics of each degree subgroup for both the top-performing baselines and the full model. All methods follow the trend of the full model and perform better with power users. Notably, \algo{}, which achieves the best overall performance, substantially mitigates the shortcomings of existing algorithms with respect to tail users. These observations indicate that substantial improvements can still be made for low-frequency users, and our framework provides a potential solution to address this challenge.

\subsection{Parameter Analysis}

In Figure~\ref{fig:vary-gamma}, we vary $\gamma$ in the range $[4, 10]$ on {\em Gowalla} , where the compression ratio parameter spans $[1/10, 1/5]$. \algo{} consistently surpasses \texttt{GraphHash} across all resolutions. Furthermore, increasing the resolution(with a corresponding increase in the embedding table size) leads to an improvement in Recall for \algo{}, whereas \texttt{GraphHash} fails to achieve further performance gains. Regarding NDCG, \algo{} maintains stable performance, while \texttt{GraphHash} demonstrates a decline, due to its coarse clustering strategy.

\subsection{Clustering Result Analysis}
\begin{figure}[!ht]
    \centering
    \begin{tikzpicture}
        \begin{customlegend}[
            legend entries={user, item, total},
            legend columns=3,
            area legend,
            legend style={at={(0.45,1.15)},anchor=north,draw=none,font=\footnotesize,column sep=0.25cm}
        ]
            \addlegendimage{color=NSCcol1, line legend, solid, very thick}
            \addlegendimage{color=NSCcol3, line legend, solid, very thick}
            \addlegendimage{color=NSCcol5, line legend, solid, very thick}
            
        \end{customlegend}
    \end{tikzpicture}
    \\[-\lineskip]
    \vspace{-4mm}
    \subfloat[{\em GraphHash}]{
        \includegraphics[width=0.33\columnwidth, trim={0 10 0 15}, clip]{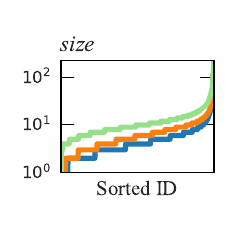}
    }
    \subfloat[{\em Leiden}]{%
        \includegraphics[width=0.33\columnwidth, trim={0 10 0 15}, clip]{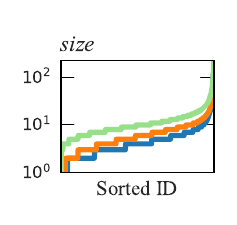}
    }
    \subfloat[{\em BACO}]{%
        \includegraphics[width=0.33\columnwidth, trim={0 10 0 15}, clip]{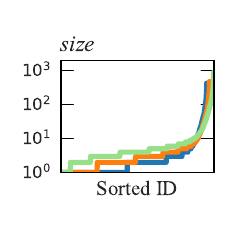}
    }
    \vspace{-2ex}
    \caption{Cluster size distributions of \texttt{GraphHash}, \texttt{Leiden}, \algo{}.}
    \label{fig:cluster-size-dist}
    \vspace{-2ex}
\end{figure}

We further examine the differences in clustering between \algo{} and strong baselines by analyzing cluster size distribution and embedding distance.

As illustrated in Figure~\ref{fig:cluster-size-dist}, \algo{} exhibits a more heterogeneous cluster size distribution compared to \texttt{GraphHash} and \texttt{Leiden}, with a greater prevalence of both small and large clusters. This pattern arises because our method groups low-degree nodes into large clusters, facilitating mutual information sharing, while assigning high-degree nodes to smaller or singleton clusters to minimize conflicts. In contrast, other methods tend to produce clusters of more uniform size, thereby failing to fully leverage the inherent graph structure.

\begin{table}[!ht]
\centering
\caption{Average distance of full embeddings and codebooks.}
\label{tab:distance}
\vspace{-2ex}
\begin{small}
\renewcommand{\arraystretch}{0.9}
\addtolength{\tabcolsep}{-0.4em}
\begin{tabular}{c|*{4}{ccc}}
\toprule
\multirow{2}{*}{\makecell{Method}} & 
\multicolumn{3}{c}{Gowalla} & 
\multicolumn{3}{c}{Yelp2018} \\
 \cmidrule(lr){2-4} \cmidrule(lr){5-7} 
\multicolumn{1}{c|}{}  & \multicolumn{1}{c}{user} & \multicolumn{1}{c}{item}  & \multicolumn{1}{c}{all} & \multicolumn{1}{c}
{user} & \multicolumn{1}{c}{item}  & \multicolumn{1}{c}{all} \\
\midrule
 \texttt{GraphHash}
& 5.444 & 4.976 & 5.174 
& 5.291 & 4.894 & 5.074 \\   

 \texttt{Leiden}
& 5.401 & 4.951 & 5.141 
& 5.449 & 4.970 & 5.188 \\   

 \texttt{SCC}
& 5.405 & 4.872 & 5.097 
& 5.038 & 4.794 & 4.905\\   

\cline{1-7}

\algo{} w/o SCU
& 5.069 & 4.909 & 4.977 
& 5.180 & 4.847 & 4.998\\

\algo{}
& 4.854 & 4.980 & 4.927 
& 4.751 & 4.982 & 4.877\\
\bottomrule
\end{tabular}
\end{small}
\end{table}

Table~\ref{tab:distance} presents the average distances between the hashing embeddings and the full model embeddings for users, items, and the combined set. The results indicate that embeddings produced by \algo{} are more closely aligned with those of the full model, corroborating its superior performance as reported in Table~\ref{tab:distance}. Moreover, the SCU strategy in \algo{} substantially reduces the user-side distance, with only minimal impact on item-side conflicts. Nevertheless, the overall improvement for both users and items is more significant.

\subsection{Additional Large-scale Datasets}
\begin{table}[!ht]
\renewcommand{\arraystretch}{0.9}
\centering
\caption{Summary statistics about large-scale datasets.}
\label{tab:large-data}
\vspace{-2ex}
\begin{small}
\begin{tabular}{c|cccc}
\hline
{\bf Dataset} & {\bf \#Users} & {\bf \#Items} & {\bf \#Interactions} & {\bf Density} \\ \hline
MovieLens & 200,808	& 65,032 & 20,228,336 & 0.155\% \\
SteamGame & 2,567,538 & 15,474 & 7,793,069 & 0.020\% \\
\hline
\end{tabular}
\end{small}
\vspace{-1ex}
\end{table}

\begin{table}[!ht]
\centering
\caption{Performance comparison on large-scale datasets.}
\label{tab:large-performace}
\vspace{-2ex}
\begin{small}
\renewcommand{\arraystretch}{0.9}
\addtolength{\tabcolsep}{-0.4em}
\begin{tabular}{c|*{2}{cccc}}
\toprule
\multirow{2}{*}{\makecell{Method}} & 
\multicolumn{3}{c}{MovieLens} & 
\multicolumn{3}{c}{SteamGame} \\
 \cmidrule(lr){2-4} \cmidrule(lr){5-7} 
\multicolumn{1}{c|}{}  & \multicolumn{1}{c}{Param\textdownarrow} & \multicolumn{1}{c}{R@20\textuparrow}  & \multicolumn{1}{c}{N@20\textuparrow} & \multicolumn{1}{c}{Param\textdownarrow} & \multicolumn{1}{c}{R@20\textuparrow}  & \multicolumn{1}{c}{N@20\textuparrow} \\
\midrule
 \texttt{Full Model}
& 17.0M & 25.980 & 20.294 
& 165.3M & 8.004 & 3.752 \\   
\cline{1-7}
 \texttt{GraphHash}
& 2.26M & 20.631 & 14.669 
& 22.6M & 6.354 & 2.821 \\   

 \texttt{Leiden}
& 2.26M & 20.878 & 14.778 
& 22.6M & 6.439 & 2.830\\    

\cline{1-7}

\algo{}
& 2.19M & 21.362 & 14.802 
& 21.7M & 6.912 & 3.111\\
v.s. Baselines
& - & +0.484 & +0.024 
& - & +0.473 & +0.281\\
v.s. Full Model
& -87.1\% & -4.618 & -5.492 
& -86.9\% & -1.092 & -0.641\\
\bottomrule
\end{tabular}
\end{small}
\end{table}

We further evaluate \algo{} on {\em MovieLens} and {\em SteamGame}, two large-scale datasets with 20M interaction edges and 2M nodes, respectively, as shown in Table~\ref{tab:large-data}.
We select the three optimal baselines, namely \texttt{GraphHash}, \texttt{Leiden}, and \texttt{SCC}, for performance evaluation on large-scale datasets. Note that since \texttt{SCC} relies on the costly SVD technique, it fails to finish running within 10 hours. Therefore, we only present the results of \texttt{GraphHash}, \texttt{Leiden}, and \algo{}. 

As shown in Table~\ref{tab:large-performace}, \algo{} overall outperforms all the baselines with fewer embedding table parameters. Compared to the best baselines, \algo{} achieves improvement of over 0.4\% in Recall@20 across both datasets. In terms of speed, \algo{} saves 25\% of the running time relative to \texttt{Leiden} on {\em MovieLens} data. On {\em SteamGame}, \algo{} runs 1.5s longer than \texttt{Leiden}, yet it achieved relative improvements of 7.3\% and 9.9\% in Recall@20 and NDCG@20, respectively. The reason lies in our compact weighting schemes, which are tailored to the recommendation task.

\end{document}